\documentclass[journal]{IEEEtran}
\usepackage{amsmath,amssymb,amsfonts}
\usepackage{amsthm}
\usepackage{calligra}
\usepackage{algorithmicx}
\usepackage[linesnumbered,ruled,vlined]{algorithm2e}
\usepackage{bm}
\usepackage{cite}
\usepackage{hyperref}
\usepackage{enumerate}
\usepackage[table]{xcolor}
\usepackage{graphicx}
\usepackage{epstopdf}
\usepackage{cite,url,epsfig}
\usepackage{epstopdf}
\usepackage{caption}
\usepackage{comment}
\usepackage{booktabs}
\usepackage{subfigure}
\newtheorem{definition}{Definition}

\begin{document}
%\title{Mechanism design for Multi-modal Semantic Communication in Mobile AIGC Networks: A Diffusion-based Stackelberg Approach}
\title{Optimizing Resource Allocation for Multi-modal Semantic Communication in Mobile AIGC Networks: A Diffusion-based Game Approach}
%title{A Diffusion-based Incentive Mechanism for Multi-modal Semantic Communications in Mobile AIGC Networks: A Stackelberg Game Approach}

\author{Jian Liu, Ming Xiao, Jinbo Wen, Jiawen Kang*, Ruichen Zhang, Tao Zhang,\\ Dusit  Niyato, \IEEEmembership{Fellow,~IEEE}, Weiting Zhang, Ying Liu

\thanks{J. Liu, M. Xiao, and J. Kang are with the School of Automation, Guangdong University of Technology, Guangzhou 510006, China (e-mail: jianliu0908@163.com, xiaoming@gdut.edu.cn, kavinkang@gdut.edu.cn). J. Wen is with the College of Computer Science and Technology, Nanjing University of Aeronautics and Astronautics, China (e-mail: jinbo1608@nuaa.edu.cn). R. Zhang and D. Niyato are with the College of Computing and Data Science, Nanyang Technological University, Singapore (e-mail: ruichen.zhang@ntu.edu.sg, dniyato@ntu.edu.sg). T. Zhang is with the School of Software Engineering, Beijing Jiaotong University, Beijing 100044, China (e-mail: taozh@bjtu.edu.cn).
Weiting Zhang and Ying Liu are with School of Electronic and Information Engineering, Beijing Jiaotong University, Beijing 100044, China (e-mail: wtzhang@bjtu.edu.cn, yliu@bjtu.edu.cn).

\textit{*Corresponding author: Jiawen Kang}
}
}

\maketitle
\begin{abstract}
%achieving a balance between the effectiveness of generation information and the complexity of the transmitted information presents a significant challenge. 

Mobile Artificial Intelligence-Generated Content (AIGC) networks enable massive users to obtain customized content generation services. However, users still need to download a large number of AIGC outputs from mobile AIGC service providers, which strains communication resources and increases the risk of transmission failures. Fortunately, Semantic Communication (SemCom) can improve transmission efficiency and reliability through semantic information processing. Moreover, recent advances in Generative Artificial Intelligence (GAI) further enhanced the effectiveness of SemCom through its powerful generative capabilities. However, how to strike a balance between high-quality content generation and the size of semantic information transmitted is a major challenge. In this paper, we propose a Generative Diffusion Model (GDM)-based multi-modal SemCom (GM-SemCom) framework. The framework improves the accuracy of information reconstruction by integrating GDMs and multi-modal semantic information and also adopts a controllable extraction module for efficient and controllable problems of unstable data recovery and slow decoding speed in GAI-enabled SemCom. Then, we introduce a novel metric called Age of Semantic Information (AoSI) based on the concept of Age of Information (AoI) to quantify the freshness of semantic information. To address the resource trading problem within the framework, we propose a Stackelberg game model, which integrates the AoSI with psychological factors to provide a comprehensive measure of user utility. Furthermore, we propose a GDM-based algorithm to solve the game under incomplete information. Compared with the traditional deep reinforcement learning algorithms, numerical results demonstrate that the proposed algorithm converges faster and is closer to the Stackelberg equilibrium.
% Compared with the traditional deep reinforcement
% learning algorithms, numerical results demonstrate that the pro-
% posed algorithm converges faster and is closer to the Stackelberg
% equilibrium.

%Specifically, the GM-SemCom segments source information into semantic data of different modalities and uses a generative decoder for efficient and accurate source information reconstruction, enabling efficient feature extraction and addressing the dynamic characteristics of GAI. We formulate a joint optimization problem in an open wireless environment, encompassing communication resources, user computing resources, and semantic information volume to optimize resource utilization while ensuring high quality and rapid response to user demands. In addition, we propose a GDM-based algorithm for optimal resource allocation. Numerical results demonstrate that the proposed schemes provide a comprehensive and efficient solution to the challenges inherent in mobile AIGC, enhancing user experiences while alleviating resource constraints.

\end{abstract}
\begin{IEEEkeywords}
Mobile AIGC, semantic communications, Stackelberg game, deep reinforcement learning, diffusion models.
\end{IEEEkeywords}
\section{Introduction}

Driven by automated information generation facilitated by Generative Artificial Intelligence (GAI) technology \cite{du2023beyond}, Artificial Intelligence-Generated Content (AIGC) has emerged as a central focus and spawned many amazing applications\cite{xu2024unleashing}, such as ChatGPT, Stable Diffusion, Sora, and Dream Fusion. Nevertheless, the effectiveness of AIGC hinges on the utilization of large neural networks containing billions of parameters, exemplified by models like GPT-4 boasting approximately 1.8 trillion parameters \cite{gpt4}. Moreover, mobile users encounter limitations stemming from constrained interactions and resource-intensive processes. Thus, the concept of mobile AIGC was proposed in \cite{xu2024unleashing}. By integrating AIGC with mobile edge computing, mobile AIGC networks can provide users with customized and real-time AIGC services on a large scale \cite{xu2024unleashing}. This integration allows a multitude of users to offload their AIGC tasks to Mobile AIGC Service Providers (MASPs), which are given sufficient physical resources to execute AIGC models. Specifically, MASPs perform AIGC inferences based on user-provided prompts\cite{wen2023freshness}, enabling users to access premium AIGC products while circumventing computational burdens and minimizing latency. 

Semantic Communications (SemCom) prioritizes the transmission of semantically meaningful information over the individual binary bits \cite{Zhang2024Wireless}. By doing so, this approach optimizes the utilization of communication resources and enhances transmission reliability, enabling efficient end-to-end delivery of high-quality content. Unlike traditional deep learning-based SemCom, which often requires network redesign for new tasks, the emergence of GAI technology overcomes this limitation. GAI-enabled SemCom can leverage GAI’s powerful generative capabilities to improve SemCom \cite{he2022robust, grassucci2023generative}. However, existing GAI-enabled SemCom frameworks remain limited by their reliance on single-modal semantic information, which prevents the full utilization of available data and overlooks the intrinsic dynamic generation capabilities of GAI models \cite{he2022robust, grassucci2023generative}. Although multi-modal approaches address the reconstruction accuracy challenge in SemCom \cite{du2024generative}, they often fail to consider the effective transmission of semantic information.

The integration of SemCom with mobile AIGC networks introduces a framework for efficiently provisioning content generation services over wireless networks, particularly in the face of constrained communication resources \cite{lin2023unified}. This framework accurately extracts relevant information from raw data, facilitating the generation of high-quality digital content for users. However, existing SemCom frameworks frequently overlook the computing resource limitations faced by mobile users and the time required for reconstructing source information. These oversights present significant challenges for the practical implementation of SemCom in real-world scenarios.

Mobile AIGC networks provide users with diverse content generation services. In latency-sensitive content generation services, the freshness of the content generation services is crucial. The delivery of fresh content generation services can enhance user immersion and significantly improve the overall Quality of Experience (QoE) of users \cite{kang2024tiny}. Therefore, in the mobile AIGC network supported by SemCom, accurate and efficient semantic information transmission is of great significance. To provide high-quality and low-latency content generation services to diverse users, it is crucial to have a large amount of communication, computing, and storage resources \cite{xu2024unleashing}. Efficient resource allocation can improve service quality and reduce resource consumption. Implementing a reasonable incentive mechanism is key to achieving optimal resource allocation, which can greatly improve the performance of mobile AIGC networks by optimizing resource utilization \cite{lai2023resource,du2023ai}. However, the specific design of incentive mechanisms in mobile AIGC networks supported by SemCom is still not fully explored in the literature. Balancing QoE and resource allocation in mobile AIGC networks supported by SemCom is a major challenge that requires further rigorous studies. We summarize several challenges existing in the current mobile AIGC networks supported by SemCom as follows:

\begin{description}
\item[\textbf{C1)}] {\textbf{Variability in GAI inferences:} Due to the dynamic characteristics of GAI inference, it is difficult for existing GAI-enabled SemCom to achieve accurate data recovery when relying on single-modal semantic information \cite{du2024generative}.}
\item[\textbf{C2)}] {\textbf{Trade-off between generation quality and transmission speed:} Transmitting more semantic information can generate higher quality AIGC products. However, this will result in slower transmission speeds due to the increased size of semantic information \cite{Zhang2024Wireless}.}
\item[\textbf{C3)}] {\textbf{Balance between resource consumption and user' QoE:} While increasing resource utilization can improve users' QoE, it also increases costs, potentially causing users to perceive the trade-off as unjustifiable \cite{kang2024tiny}. Therefore, effectively implementing resource allocation schemes remains a significant challenge.}
\end{description}

To address the above challenges, we propose a GDM-based multi-modal SemCom (GM-SemCom) framework in mobile AIGC networks. The framework consists of an extraction module and GDM design to achieve controllable semantic information size transmission, allowing MASPs to optimize semantic information extraction according to user needs to balance transmission speed and reconstruction accuracy. Likewise, the users can dynamically control the number of steps for generating GDM-based semantic decoders. In addition, the framework uses multi-modal technology to fully extract semantic information to achieve accurate transmission and accelerate the execution of GDM-based semantic decoders.

Moreover, to ensure high-quality AIGC services between the MASP and users within our proposed framework, we formulate a Stackelberg game model with the MASP as the leader and multiple users as followers. Here, Age of Information (AoI) is only a generic metric of information freshness \cite{zhang2023learning}. Thus, we propose a new metric called Age of Semantic Information (AoSI) based on the concept of AoI to measure the efficiency of AIGC services. Then, we incorporate AoSI and psychological factors to create the Stackelberg game model to improve QoE and achieve efficient resource allocation. Furthermore, the Stackelberg game model is a multi-stage competitive game under incomplete information conditions \cite{jiang2022reliable}. To obtain the optimal Stackelberg game solution, we propose a GDM-based Deep Reinforcement Learning (DRL) algorithm, where we use a diffusion model to generate actor policy in DRL. The main contributions of this paper are summarized as follows: 
\begin{itemize}
%\item {We present the GM-SemCom framework for mobile AIGC. With the support of GM-SemCom, MASP only needs to transmit the compressed semantic information of text prompts and visual prompts of AIGC products, On the user side, a lightweight decoder is employed to dynamically manage the decoding process. This results in users being able to obtain high-quality AIGC products while reducing both computational complexity and energy costs.}
%\item {To achieve accurate reconstruction of the source message for Mobile AIGC networks, we propose a GM-SemCom framework by integrating multi-modal prompts and diffusion models. This framework can effectively address the inherent challenge of unstable data recovery that is common when employing GAI models with single-modal prompts in SemCom. (for \textbf{C1})}
\item {We propose a novel GM-SemCom framework for mobile AIGC  by integrating diffusion models and multi-modal semantic information. This framework can achieve controllable extraction of semantic information and effectively solve the problems of unstable data recovery and slow decoding speed in GAI-enabled SemCom (for \textbf{C1} and \textbf{C2}). }
% We propose a novel GM-SemCom framework for mobile AIGC by integrating diffusion models and multi-modal semantic information. This framework can achieve controllable transmission of semantic information size and effectively solve the problems of unstable data recovery and slow decoding speed in GAI-enabled SemCom.
\item {To enhance the measure of the QoE of AIGC services, we proposed a new metric AoSI considering specifically the freshness of semantic information extraction and transfer. In addition, to improve the utility of both the MASP and users in the framework and achieve efficient resource allocation, we formulate a Stackelberg game model and incorporate AoSI and psychological factors to create comprehensive user utility (for \textbf{C3}).}
\item{We design a GDM-based DRL algorithm to find Stackelberg equilibrium under incomplete information conditions, in which a diffusion model is used to explore each round of the Stackelberg game, which enables accurate and robust convergence to equilibrium in complex and uncertain environments (for \textbf{C3}).}
\end{itemize}
%介绍章节

The remainder of this paper is structured as follows: In Section \ref{related}, we review related literature. In Section \ref{system_model}, we introduce the system model. In Section \ref{framework}, we introduce the architecture of GM-SemCom. In Section \ref{problem}, we introduce the Stackelberg game for resource allocation in mobile AIGC networks. Section \ref{GDM} introduces the GDM-based DRL algorithm to solve the game under incomplete information conditions. Section \ref{results} evaluates the proposed schemes. Section \ref{conclusion} provides the conclusion of this paper. 

\section{Related Work}\label{related}
\subsection{Mobile AIGC Networks}
% 先介绍Mobile AIGC, 再介绍Resource allocation在Mobile AIGC
% 
% GAI technology encapsulates a range of methodologies, notably including Variational Autoencoders (VAEs), Generative Adversarial Networks (GANs), Transformers, and GDMs \cite{lai2023resource}, all of which enable autonomous content creation across various modalities. Notably, GDMs combine high-quality generation with training stability and possess good controllability and robustness.
GAI techniques cover a range of methods such as Variational Autoencoders (VAEs), Generative Adversarial Networks (GANs), Transformers, and GDMs \cite{lai2023resource}. Users initially access cloud-based AIGC services over the core network by executing GAI models on cloud servers. However, cloud services exhibit high latency due to the remote nature of users. Hence, mobile AIGC emerges as a field leveraging GAI technologies to generate, process, and distribute content directly on mobile devices or at the edge\cite{xu2024unleashing}. Researchers have rapidly promoted the development of mobile AIGC from both model and system aspects \cite{xu2024unleashing}, concentrating on optimizing AIGC models and managing mobile AIGC networks. For instance, the authors in \cite{du2023exploring} proposed a collaborative distributed diffusion-based AIGC framework, facilitating the effective execution of AIGC tasks and optimizing edge computation resource utilization. The authors in \cite{wen2023freshness} first utilized the AoI metric to quantify data freshness for AIGC fine-tuning. In \cite{liao2024optimizing}, the authors proposed a graph attention network-based information propagation optimization framework, which utilizes blockchain technology to better manage AIGC products securely and avoid the tampering and plagiarization of AIGC products.

Although AIGC could potentially revolutionize existing production processes, users who currently access AIGC services through mobile devices lack support for resource-intensive data generation services and there are resource trading issues between users and MASPs. In this paper, we propose the GM-SemCom framework for mobile AIGC networks and formulate an incentive mechanism to address these issues.

\subsection{Incentive Mechanisms in Mobile AIGC Networks}

In the context of mobile AIGC, incentive mechanisms play a key role in attracting participants to contribute resources and data to AIGC services. In \cite{lai2023resource}, the authors proposed a resource-efficient incentive mechanism framework for mobile AIGC networks in resource-constrained environments. They use model partitioning to offload AI tasks efficiently and apply a Stackelberg game model to motivate edge devices to contribute computing resources. In the mobile AIGC vehicle networking scenario, to improve the unloading efficiency and maximize the utility of drones and vehicles, the authors in \cite{dai2021vehicle} modeled the computing data transaction process between drones and vehicles as a bargaining game. In addition, the authors in \cite{xu2022wireless} designed a double Dutch auction mechanism to determine the best pricing and allocation rules in this market, allowing VR users (i.e., buyers) and VR service providers (i.e., sellers) to quickly trade VR services. In \cite{du2023ai}, the authors designed a contract-theoretic AI-generated incentive mechanism to promote semantic information sharing among users in the mobile AIGC metaverse service. In \cite{ismail2022semantic}, the authors introduced a semantic extraction algorithm to reduce data transmission over wireless channels. Additionally, they proposed a reverse auction mechanism that enables service providers to select devices capable of enhancing virtual object quality through semantic information.

Although there are many incentive mechanisms for resource trading, current research often ignores the importance of semantic information size and data freshness in the valuation function. We incorporate semantic information size and data freshness into our proposed Stackelberg game to solve resource optimization challenges in mobile AIGC networks.

\subsection{Semantic Communications}
% 列举DL-semcom models，把你introduction这部分的内容挪到这里来，但目前还没有人结合GAI
%In the era of rapidly expanding data volumes, the limitations of classical communication theory are becoming increasingly evident. 
% SemCom represents an innovative paradigm where message transmission transcends mere content, involving the direct extraction of relevant semantic information. This approach eliminates redundant data and mitigates associated costs. Deep learning has exhibited remarkable efficacy in feature extraction, garnering widespread application in the field of SemCom, e.g., channel state information compression \cite{yin2022deep} and end-to-end communication systems \cite{ye2020deep}. In \cite{bourtsoulatze2019deep}, The authors designed a deep learning-based joint source-channel coding (JSCC) SemCom system that can accurately transmit even in low signal-to-noise (SNR) ratio conditions. Subsequent studies explored deep learning-based JSCC transmission schemes for various data types including text\cite{xie2021deep, peng2022robust}, image\cite{kurka2021bandwidth,hu2023robust}, speech \cite{han2022semantic}, and multi-modal data \cite{xie2022task}. 
SemCom represents a paradigm shift in transmission, focusing on extracting relevant semantic information to reduce redundant data and associated costs \cite{Zhang2024Wireless}. Deep learning has proven highly effective in this area. In \cite{bourtsoulatze2019deep}, The authors designed a deep learning-based Joint Source Channel Coding (JSCC) SemCom system that can accurately transmit even in low Signal-to-Noise Ratio (SNR) conditions. Subsequent studies explored deep learning-based JSCC transmission schemes for various data types including text\cite{xie2021deep}, image\cite{kurka2021bandwidth}, speech \cite{han2022semantic}, and multi-modal data \cite{xie2022task}. The studies of SemCom are broadly categorized into two groups based on their communication objectives, i.e., data reconstruction and task inference. Actually, contemporary data traffic is predominantly visual, with video and image data comprising 75\% of IP traffic, transmitting this multimedia content faces challenges like latency and network congestion. 
%To address these demands, the development of SemCom that reduces data size without compromising semantic fidelity is essential.

In the real world, resources in wireless networks, including transmission power, GPU/CPU cores, and computational capabilities, typically exhibit dynamic availability. Most existing SemCom models rely on traditional deep learning \cite{bourtsoulatze2019deep, xie2021deep, kurka2021bandwidth, han2022semantic, xie2022task}. Direct integration of deep learning-based SemCom with mobile AIGC faces challenges due to fixed encoder and decoder structures, resulting in elevated semantic noise and decoding complexity. Completely redesigning these structures is incompatible with existing communication networks. Fortunately, GAI technology shows promise in complementing SemCom. For example, the authors in \cite{he2022robust} proposed a neural network-based image transmission system trained by the GAN to achieve robust transmission. In \cite{grassucci2023generative}, a diffusion model was used to synthesize semantically consistent scenes by spatially adaptive normalizing denoised semantic information. The authors in \cite{du2024generative} solved the reconstruction accuracy problem in GAI-enabled SemCom through multi-modality without considering the efficient transmission of semantic information, causing greater overhead than traditional communication. 

Existing GAI-enabled SemCom focuses on the accuracy of information reconstruction, but there is limited research on the size of semantic information data transmitted. More importantly, to provide users with uninterrupted immersive AIGC services, information freshness is very important. Therefore, our focus is on how to transmit semantic information and achieve accurate information reconstruction efficiently.

\begin{figure*}
%\vspace{-0.5cm}
\centering
\includegraphics[width=0.9\textwidth]{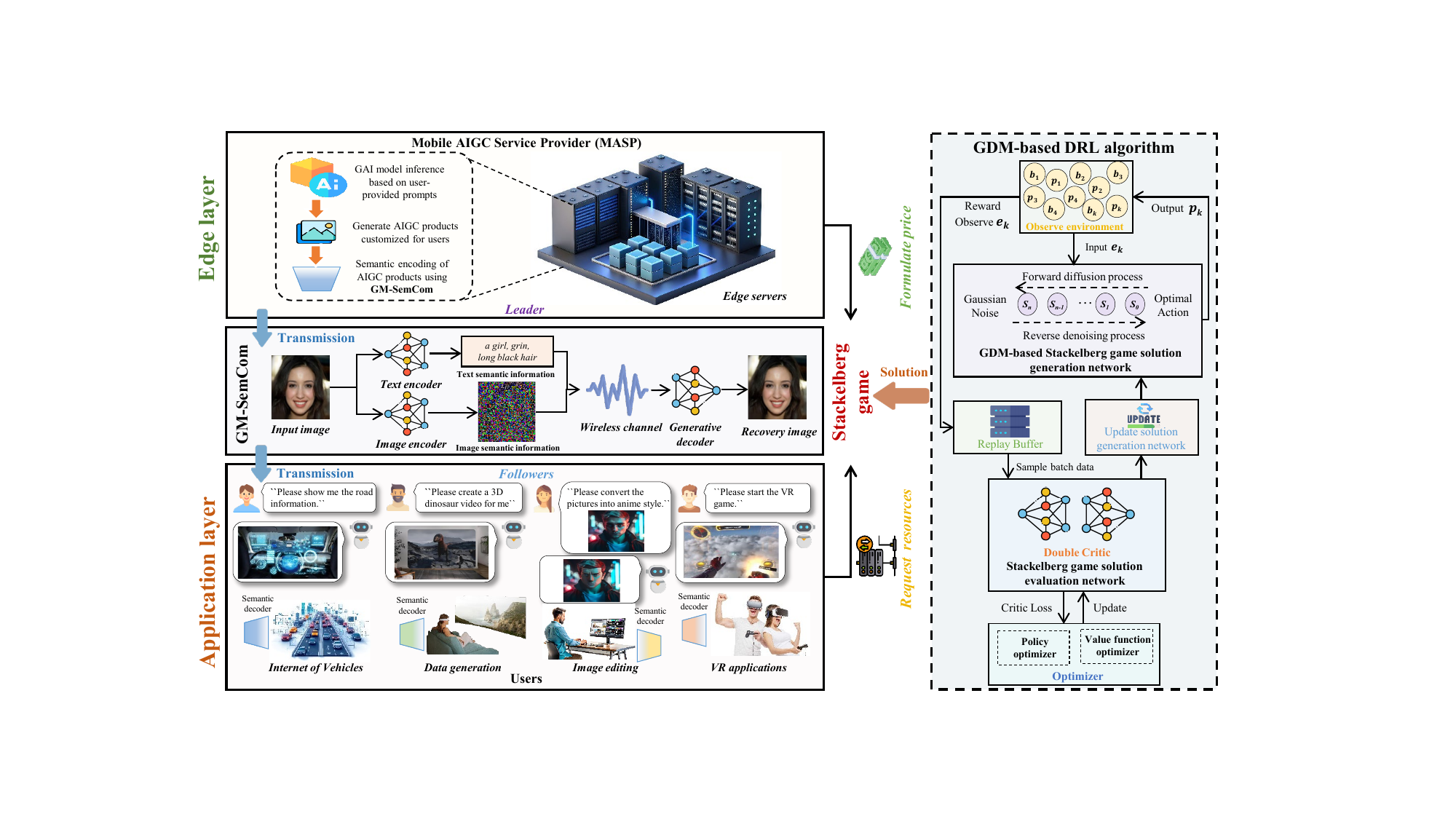}
\captionsetup{font=footnotesize}
\caption{The GM-SemCom framework for mobile AIGC networks. Specifically, we introduce several examples of vision-based mobile AIGC services supported by GM-SemCom. Then, based on the framework, we model the resource trading problem as a Stackelberg game with incomplete information. Finally, we propose a GDM-based DRL algorithm to solve the optimal solution of the Stackelberg game with incomplete information.}
\label{system}
\end{figure*}

\section{System Model}\label{system_model}
%In this section, we give the formal introduction of our proposed  GM-SemCom framework for mobile AIGC. 
Without loss of generation, we consider a typical scenario with a set $\mathcal{M} = \{1,\ldots, i,\ldots, M\}$ of $M$ users and one MASP in mobile AIGC networks, as illustrated in Fig. \ref{system}. Following this scenario, It is not hard to extend our scenario to multiple MASPs in future work. In this paper, we focus on one MASP to simplify the study and concentrate on the core challenges within the SemCom framework. To acquire high-quality AIGC services, users first send their prompts to the MASP. Once the MASP receives prompts from the users, it utilizes various GAI models based on the requirements and provides AIGC services for the users. Then, the proposed GM-SemCom framework can be applied to overcome the communication resource constraints and efficient transmission. Specifically, the MASP performs GAI model inference after accepting user-provided prompts. For example, drivers request vision-based path planning in autonomous driving, and the drivers send prompts to the MASP. Then, the MASP performs the inferences by executing GAI models based on the driver's prompts. During the inference process, the GAI models are used to generate the path planning according to the needs of drivers. The generated path planning is then transmitted back to the drivers. However, accurate information reconstruction is challenging for GAI-enabled single-modal SemCom due to the dynamic inference characteristics of GAI. To accurately transmit the generated AIGC products and overcome the communication resource limitations, a novel multi-modal semantic encoder is used, where the AIGC products are encoded into extracted semantic information and transmitted over a wireless channel. Finally, the users use a generative semantic decoder to recover AIGC products using the received semantic information. We design the proposed framework to provide tailored services for users, aligned with their distinct requirements across varied scenarios.

% To exemplify the practical applications, we present several scenarios:
% \begin{itemize}
% \item \textbf{IoV:} The MASP, through integration with eXtended Reality (XR) technology, offers advanced services to enhance driving safety and in-cockpit experiences \cite{huang2023service}. These include driving guidance, route optimization, safety warnings, and in-vehicle metaverse gaming \cite{zhang2024generative}.
% \item \textbf{Image Editing:} The MASP provides sophisticated capabilities like high-precision image super-scaling, comprehensive image restoration, and advanced picture augmentation \cite{lai2023resource}. These features optimize real-time visual content modification and enhancement.
% \item \textbf{Data Generation:} The MASP is capable of generating data tailored to user-specific requirements. This functionality is instrumental in supporting scientific research by providing customized datasets, as well as addressing various everyday needs with personalized services, thereby enhancing both academic and practical applications.
% \item \textbf{Virtual Reality:} The MASP supports efficient content production with minimal latency, meeting the demands of real-time environments \cite{du2023ai}. Apple~Vision~Pro exemplifies successful AIGC deployment, enhancing user experience through real-time image analysis, object detection, and augmented visuals.
% \end{itemize}

\subsection{GM-SemCom Framework Design}

Figure \ref{Framework_1} shows the proposed architecture of GM-SemCom. On the encoding side, GM-SemCom uses a text encoder and an image encoder to extract semantic information from input \cite{du2024generative}. The objective of the text semantic encoder is to extract the information from the input image $\mathbf{i}_0$ into a descriptive vector $\mathbf{t}_\mathrm{sem}$, which contains the essential details of text semantic information for assisting the decoder. $\mathbf{t}_{\mathrm{sem}}$ can be expressed as $\mathbf{t}_{\mathrm{sem}} = \mathcal{E}_t(\mathbf{i}_0;\boldsymbol{\zeta}),$ where $\mathcal{E}_t(\cdot;\zeta)$ represents the text semantic encoder with a trainable parameter set $\boldsymbol{\zeta}$. The image semantic encoder encodes the image into image semantic information $\mathbf{i}_\mathrm{sem}$ to capture image structure information \cite{du2024generative}. \(\mathbf{i}_{\mathrm{sem}}\) can be expressed as $\mathbf{i}_{\mathrm{sem}} = \mathcal{E}_v(\mathbf{i}_0; \boldsymbol{\nu}),$ where $\mathcal{E}_v(\cdot;\nu)$ represents the visual semantic encoder with a trainable parameter set $\boldsymbol{\nu}$. The semantic information is transmitted through channel coding, and the transmission signal $\mathbf{T}_\mathrm{s}$ is obtained, which is given by $\mathbf{T_\mathrm{s}} = \mathcal{E}_c(\mathbf{t}_{\mathrm{sem}}, \mathbf{i}_{\mathrm{sem}}; \boldsymbol{\varphi}),$ where $\mathcal{E}_c(\cdot,\cdot;\boldsymbol{\varphi})$ represents channel encoder with a trainable parameter set $\boldsymbol{\varphi}$. The compression rate is defined as the ratio of the total bit length of the transmitted information to the bit length of the original uncompressed image, which is given by $\mathbf{R} = \frac{\mathcal{L}(\mathbf{T_\mathrm{s}})}{\mathcal{L}(\mathbf{i}_0)},$ where $\mathcal{L}$ represents the bit length of the transmitted semantic information. The received signal $\mathbf{R_s}$ can be represent as $\mathbf{R_\mathrm{s}}=\mathbf{H}\cdot\mathbf{T_\mathrm{s}}+\mathbf{N_\mathrm{p}},$ where $\mathbf{H}$ represents the physical channel gain. Without loss of generality, we consider the physical impairment caused by the Rayleigh fading channel and $\mathbf{N}_\mathrm{p}\sim\mathcal{CN}(0,\sigma_{n}^{2})$. The received multi-modal semantic information $\hat{\mathbf{t}}_{\mathrm{sem}}$ and $\hat{\mathbf{i}}_{\mathrm{sem}}$ are recovered after passing through the channel decoder and can be expressed as $\hat{\mathbf{t}}_{\mathrm{sem}},\hat{\mathbf{i}}_{\mathrm{sem}}= \mathcal{D}_c(\mathbf{R_x};\boldsymbol{\iota}),$ where $\mathcal{D}_c(\cdot;\boldsymbol{\iota})$ represents the channel decoder with a trainable parameter set $\boldsymbol{\iota}$. The decoder uses text semantic information and image semantic information to accurately reconstruct the source image expressed as $\Hat{\mathbf{x}_0} = \mathcal{D}_s(\mathbf{t}_{\mathrm{sem}}, \mathbf{i}_{\mathrm{sem}}; \boldsymbol{\beta}),$ where $\mathcal{D}_s(\cdot,\cdot;\boldsymbol{\beta})$ represents semantic decoder with a trainable parameter set $\boldsymbol{\beta}$. Finally, the users obtain AIGC output through GM-SemCom.

\begin{figure*}
%\vspace{-0.5cm}
\centering
\includegraphics[width=0.9\textwidth]{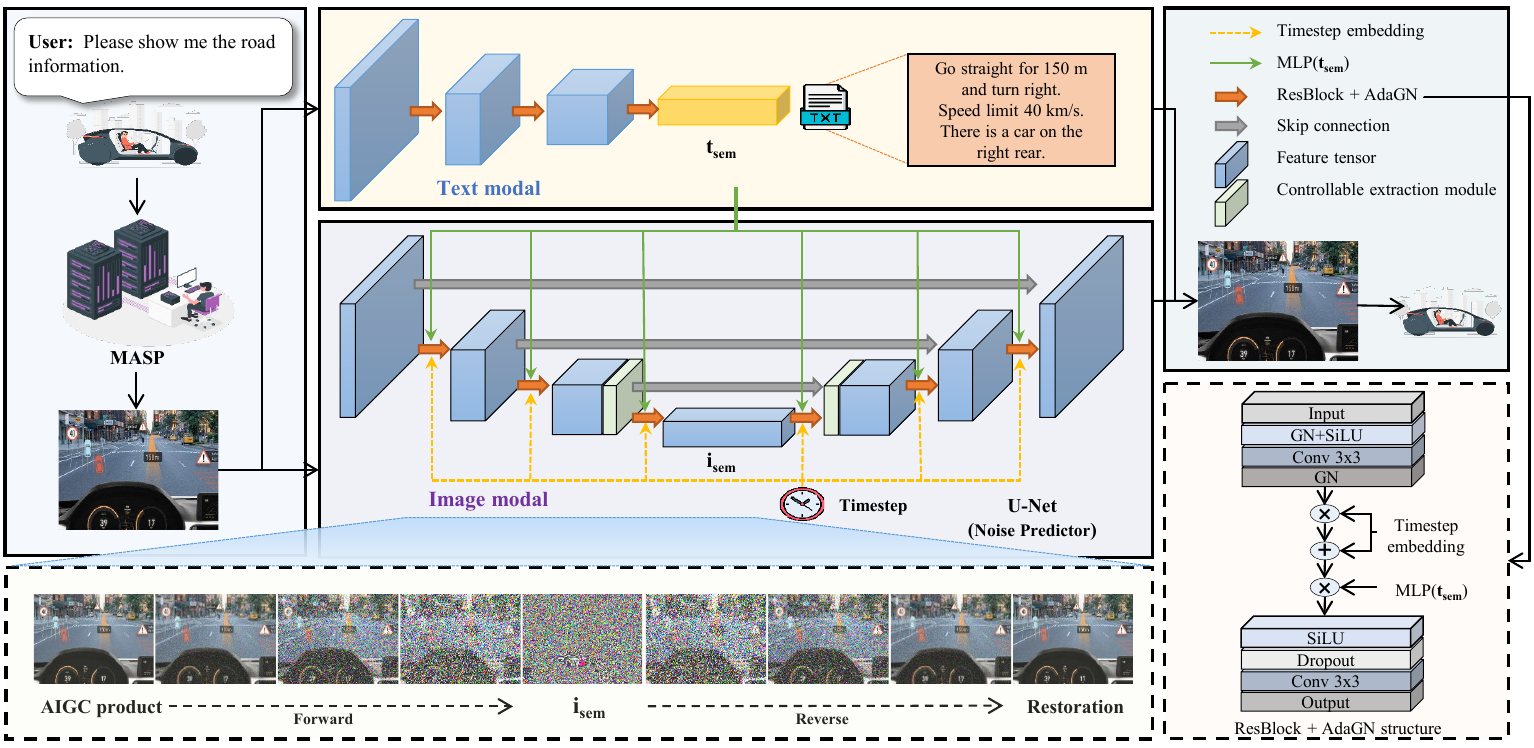}
\captionsetup{font=footnotesize}
\caption{The architecture of GDM-based multi-modal SemCom.}
\label{Framework_1}
\end{figure*}
\subsection{Stackelberg Game for Resource Allocation}

In the proposed GM-SemCom framework for mobile AIGC networks, the MASP is the resource holder, and users utilize the MASP's bandwidth resources to achieve GM-SemCom and obtain AIGC services. Therefore, a monopoly market for resource trading emerges, in which the MASP, acting as a monopolist, exercises pricing power over bandwidth. The users then must decide the amount of bandwidth to purchase based on the price to achieve AIGC services \cite{zhang2023learning,kang2024tiny}. Specifically, when the price of bandwidth decreases, users may choose to purchase additional bandwidth to enhance transmission speed and improve their QoE. On the contrary, when the price is high, users are unwilling to buy more bandwidth.

To maximize the profit of the MASP, we propose a Stackelberg game model to achieve strategic price setting.  Specifically, the Stackelberg game between the MASP and users consists of two stages. In the first stage, the MASP serves as the leader and sets the price for selling bandwidth to optimize its utility. In the second stage, each user acts as a follower and determines their bandwidth demands to optimize the QoE. Furthermore, under incomplete information conditions, it is difficult for traditional DRL algorithms to solve the Stackelberg game, i.e., slowly converge to equilibrium or easily fall into the local optimum\cite{du2023ai}. To address the above problems, we propose a GDM-based DRL algorithm to generate the optimal Stackelberg game solution.

\section{GDM-based Multi-modal SemCom with Controllable extraction}\label{framework}
In this section, we introduce the framework of GM-SemCom in mobile AIGC networks, whose core is a conditional semantic diffusion model and a controllable compression module. 

\subsection{Conditional Semantic Diffusion Models}
The core of the conditional semantic diffusion model is the Denoising Diffusion Implicit Model (DDIM) \cite{song2020denoising}, which is characterized by a deterministic generation process
\begin{equation}\mathbf{i}_{t-1}=\sqrt{\alpha_{t-1}}\left(\frac{\mathbf{i}_{t}-\sqrt{1-\alpha_{t}}\epsilon_{\theta}^{t}(\mathbf{i}_{t})}{\sqrt{\alpha_{t}}}\right)+\sqrt{1-\alpha_{t-1}}\epsilon_{\theta}^{t}(\mathbf{i}_{t}),
\label{function1}
\end{equation}
and the following inference distribution
\begin{equation}\label{inference_distribution}
q(\mathbf{i}_{t-1}|\mathbf{i}_t,\mathbf{i}_0)=\mathcal{N}\bigg(\sqrt{\alpha_{t-1}}\mathbf{i}_0+\sqrt{1-\alpha_{t-1}}\frac{\mathbf{i}_t-\sqrt{\alpha_t}\mathbf{i}_0}{\sqrt{1-\alpha_t}},\mathbf{0}\bigg).
\end{equation} 

To accurately reconstruct the source information, we design a conditional semantic decoder based on the conditional DDIM \cite{preechakul2022diffusion}. The decoder uses the text semantic information ${\mathbf{t}}_\mathrm{sem}$ as a condition to guide $\mathbf{i}_\mathrm{sem}$ to denoise and generate the source image \cite{du2024generative}. This design enables accurate reconstruction of the source information while accelerating decoding. Combining the text semantic encoder and the inverse of (\ref{function1}), $\mathbf{i}_{t+1}$ is
\begin{equation}\label{x_{t+1}}
\mathbf{i}_{t+1}=\sqrt{\alpha_{t+1}}\mathbf{f}_{\theta}(\mathbf{i}_{t},t,\mathbf{t}_{\mathrm{sem}})+\sqrt{1-\alpha_{t+1}}\epsilon_{\theta}(\mathbf{i}_{t},t,\mathbf{t}_{\mathrm{sem}}).
\end{equation}

The conditional semantic decoder takes two modal semantic information (i.e., ${\mathbf{t}}_\mathrm{sem}$ and ${\mathbf{i}}_\mathrm{sem}$) as input \cite{du2024generative}. Conditional semantic decoder models $p_{\theta} (\mathbf{i}_{t-1}|\mathbf{i}_t, {\mathbf{t}}_\mathrm{sem})$ to match the inference distribution \( q(\mathbf{i}_{t-1}|\mathbf{i}_t, \mathbf{i}_0) \) defined in (\ref{inference_distribution}), with the following reverse generative process:
\begin{equation}p_\theta(\mathbf{i}_{0:T}\mid\mathbf{t}_{\mathrm{sem}})=p(\mathbf{i}_T)\prod_{t=1}^Tp_\theta(\mathbf{i}_{t-1}\mid\mathbf{i}_t,{\mathbf{t}}_\mathrm{sem}),
\label{generative process}
\end{equation}
which can be further expressed as 
\begin{equation}\label{reverse generative pro}
p_\theta(\mathbf{i}_{t-1}|\mathbf{i}_t,{\mathbf{t}}_\mathrm{sem})=\begin{cases}\mathcal{N}(\mathbf{f}_\theta(\mathbf{i}_1,1,{\mathbf{t}}_\mathrm{sem}),\mathbf{0}),&\text{if }  t=1,\\q(\mathbf{i}_{t-1}|\mathbf{i}_t,\mathbf{f}_\theta(\mathbf{i}_t,t,{\mathbf{t}}_\mathrm{sem})),&\text{otherwise}.
\end{cases}\end{equation}

We parameterize $\mathbf{f}_{\theta}$ in \eqref{reverse generative pro} as a noise prediction network $\epsilon_{\theta}(\mathbf{i}_{t},t,\mathbf{t}_{\mathrm{sem}})$ \cite{song2020denoising}, given by
\begin{equation}
\mathbf{f}_{\theta}(\mathbf{i}_{t},t,\mathbf{t}_{\mathrm{sem}})=\frac{1}{\sqrt{\alpha_{t}}}\left(\mathbf{i}_{t}-\sqrt{1-\alpha_{t}}\epsilon_{\theta}(\mathbf{i}_{t},t,\mathbf{t}_{\mathrm{sem}})\right).
\label{function4}
\end{equation}

We use the Adaptive Group Normalization layer (AdaGN) to constrain U-Net \cite{dhariwal2021diffusion}, extending group normalization by applying channel-wise scaling and shifting on the normalized feature map $\mathbf{h}\in\mathbb{R}^{c\times h\times w}$. The AdaGN condition is
\begin{equation}\mathrm{AdaGN}(\mathbf{h},t,\mathbf{t}_{\mathrm{sem}})=\mathbf{z}_s(\mathbf{t}_s\text{GroupNorm}(\mathbf{h})+\mathbf{t}_b),\end{equation}
where $\mathbf{z}_s \in \mathbb{R}^c = \mathrm{Affine}(\mathbf{t}_{\mathrm{sem}})$ and $(\mathbf{t}_s, \mathbf {t}_b) \in \mathbb{R}^{2 \times c} = \mathrm{MLP}(\psi(t))$ represents the output of the multi-layer perception using the sinusoidal encoding function $\psi$ as Timestep embedding \cite{song2020denoising}. These layers are used consistently throughout the U-Net architecture.

Training is done by optimizing the loss function $L_{\mathrm{simple}}$ \cite{du2024generative} with respect to $\theta$, given by
\begin{equation}
L_{\mathrm{simple}}=\sum_{t=1}^T\mathbb{E}_{\mathbf{i}_0,\epsilon_t}\Big[\left\|\epsilon_\theta(\mathbf{i}_t,t,\mathbf{t}_{\mathrm{sem}})-\epsilon_t\right\|_2^2\Big],
\label{function7}
\end{equation}
where $\epsilon_{t}\in\mathbb{R}^{3\times h\times w}\sim\mathcal{N}(\mathbf{0},\mathbf{I})$, $\mathbf{i}_{t}=\sqrt{\alpha_{t}}\mathbf{i}_{0}+\sqrt{1-\alpha_{t}}\epsilon_{t}$.

% $\mathbf{x}_{T}$ can be construed as the visual semantic information $\mathbf{v}_\mathrm{sem}$ of the image $\mathbf{x}_{0}$.
% \subsection{Semantic Encoder and Controllable Compression Modul}

\subsection{Controllable Extraction Module}
To achieve efficient transmission of semantic information, we introduce a semantic extraction layer at the end of the image semantic encoder as a controllable extraction module, which is implemented based on the image compression method \cite{fadnavis2014image}. This module uses the bicubic convolution algorithm to determine the grayscale value based on the weighted average of the nearest 16 pixels of the given input and assigns this value to the compressed image semantic information $\mathbf{i}_\mathrm{sem}$ as output. The bicubic convolution interpolation kernel is \cite{keys1981cubic} 
\begin{equation}
u\left(s\right)=\begin{cases}\frac{3}{2}\left|s\right|^{3}-\frac{5}{2}\left|s\right|^{2}+1&0<\left|s\right|<1,\\
-\frac{1}{2}\left|s\right|^{3}+\frac{5}{2}\left|s\right|^{2}-4\left|s\right|+2&1<\left|s\right|<2,
\\0&2<\left|s\right|,
\end{cases}
\end{equation}
where $s$ is the distance between the point and nearby semantic information pixels.

The use of this module can not only regulate the image semantic information to achieve efficient transmission of information but also ensure that the compressed image can still retain key image semantic information during the reconstruction process. In addition, the controllable compression module enables the MASP to provide specific services based on the user's pursuit of transmission speed or image quality. For example, if a user needs a task-oriented SemCom with low image quality and low latency, a higher compression rate can be controlled to increase the transmission speed.
% In Section \ref{effetiveness of GM}, we compare the reconstruction performance of GM-SemCom with GAI-enabled SemCom and traditional communication at the same compression rate. 

The controllable extraction module's design can adjust the data size to be transmitted, thereby increasing the transmission speed and reducing the bandwidth required. However, considering the limited bandwidth, it is necessary to balance bandwidth consumption and users' QoE. Hence, in the next section, we propose a Stackelberg game approach for efficient resource allocation in mobile AIGC networks.

\section{Stackelberg Game for Resource Allocation in Mobile AIGC Networks}\label{problem}
In this section, we first propose a new metric called AoSI based on the concept of AoI, which can evaluate the transmission speed of semantic information. Then, we design a Stackelberg game model to ensure high-quality AIGC services and prove the existence and uniqueness of the equilibrium between the MASP and users  \cite{zhang2023learning}.

\subsection{Age of Semantic Information}
AoI is a metric measuring the freshness of information, which is defined as the time since the last data update event. Applying the AoI is expected to improve the performance of time-based services \cite{liew2023mechanism}. To quantify the freshness of SemCom tasks in mobile AIGC services, we propose the AoSI based on the concept of AoI, which is defined as the time from the last time the user successfully received the semantic information to the current update.

We apply Orthogonal Frequency Division Multiple Access (OFDMA) technology to wireless communication between MASP and users \cite{wen2023freshness}. Specifically, each user is assigned a downlink orthogonal Resource Block (RB). According to \cite{wang2022performance}, the downlink channel capacity corresponding to $i^{th}$ is expressed as
\begin{equation} \label{OFDMA}
C_i = \log_{2}\left(1+\frac{{\gamma d_i^{-\epsilon}} P_i}{N_{0}}\right), i \in \{1, 2, \ldots, M\},
\end{equation}
where $P_i$ represents the transmission power of the MASP, $N_0$ denotes the average noise power, $\gamma$ is the unit channel power gain, $d_i$ is the physical distance, and $\epsilon$ is the path-loss coefficient. The transmission rate of the MASP to user $i$ is given by \cite{zhang2023learning}
\begin{equation}
r_i=b_iC_i,
\end{equation}
where $b_i$ is denoted as the bandwidth demand requested by user $i$ to the MASP. To simulate the trade-off between update intervals and freshness of semantic information, we use per-packet AoI \cite{liew2023mechanism} to reflect the performance of AoSI. Thus, for user $i$, the AoSI of SemCom task is
\begin{equation}\label{AoSI_formulate}
A_i=\frac{V_i}{r_i},
\end{equation} 
where $V_i$ is the number of bits of semantic information to be transmitted, calculated by $V_i=R_i{\mathcal{L}_i(\mathbf{x}_0)}$. $R_i$ represents the compression rate of SemCom between the MASP and user $i$ and
${\mathcal{L}_i(\mathbf{x}_0)}$ represents the bit length of the source information.

\subsection{Stackelberg Game between the MASP and Users}\label{Stackelberg Game between MAPS and users}
We propose a one-leader and multiple-followers Stackelberg game, where the MASP takes the role of leader and users act as followers. The users can access the MASP to request various AIGC services. Both the MASP and users optimize their utility by adjusting their strategies. The Stackberg game can be decomposed into the following two sub-games.

\textbf{1) Bandwidth demands in Stage II:} The user's utility is the difference between the benefits of obtaining AIGC services and the cost of purchasing bandwidth. A higher AoSI causes high latency, resulting in a decrease in the user's QoE. In addition to the latency for AIGC services, we also consider the service quality of the MASP transmitted through GM-SemCom. We use the Structural Similarity Index Measurement (SSIM) \cite{wang2004image} to measure the quality of the transmitted image. SSIM is defined as
\begin{equation}\label{SSIM}
SSIM(x,y)=\frac{(2\mu_{x}\mu_{y}+C_{1})(2\sigma_{xy}+C_{2})}{(\mu_{x}^{2}+\mu_{y}^{2}+C_{1})(\sigma_{x}^{2}+\sigma_{y}^{2}+C_{2})},
\end{equation}
where $\mu_{x}$, $\sigma_{x}$, and $\sigma_{xy}$ are the mean, standard deviation, and cross-correlation between the two images $x$ and $y$, respectively. The terms $C_{1}$ and $C_{2}$ can avoid instability when the means and variances are close to zero. 

The Weber-Fechner law describes the relationship between perceived magnitude and stimulus intensity. Adopting the Weber-Fechner law in conjunction with the metaverse immersion metric \cite{du2023attention}, we use logarithms to model the QoE of the users. Therefore, combined with AoSI (\ref{AoSI_formulate}), the utility function of user $i$ can be expressed as
\begin{equation}
U_i(b_i) =  \delta_{i} \ln \left(1 + 1/A_i\right) \ln(1+SSIM_{i})  - pb_i ,
\end{equation}
where $\delta_{i}$ is the immersion coefficient of user $i$ and $p>0$ represents the monetary payment per unit of bandwidth resource of the MASP used by user $i$. During Stage II, each user seeks to maximize its utility $U_i(b_i)$ by selecting optimal bandwidth to purchase given the bandwidth price set by the MASP. Consequently, the optimization problem for user $i$ can be formulated as 

\begin{equation}
\begin{aligned}
\textbf{Problem 1:} & \quad \max_{b_i} \quad U_n(b_i) \\
& \quad ~\text{s.t.} \quad b_i > 0.
\end{aligned}
\end{equation}

\textbf{2) Bandwidth sales prices in Stage I:} The MASP's utility is the difference between the bandwidth fee paid by all users and the bandwidth resource cost, which is affected by the bandwidth price and user bandwidth demands \cite{zhang2023learning}. Hence, the utility of the MASP is
\begin{equation}\label{utilty_MASP}
U_s(p) = \mathcal{P} + \sum_{i=1}^{M} (p \cdot b_i - c \cdot b_i),
\end{equation}
where $ \mathcal{P}>0$ is the AIGC service fee, and $c > 0$ is the unit transmission cost of bandwidth for transmitting AIGC products. In the first stage, the MASP aims to maximize its benefit by setting a pricing strategy under the constraints of maximum allowable bandwidth $B_{\rm{max}}$ and a unit price cap $p_{\rm{max}}$. The MASP must ensure that the total bandwidth sold does not exceed $B_{\rm{max}}$ and that the unit price remains below $p_{\rm{max}}$. Consequently, the optimization of the MASP's utility can be formulated as
\begin{equation}
\begin{aligned}
\textbf{Problem 2:}\quad \max_{p} \quad U_s(p) \quad\quad\quad\quad\quad\quad\quad\\
\quad\quad\quad \text{s.t}.
\begin{cases}
0 < \sum_{i=1}^{M} b_i \leq B_{\rm{max}}, \\
b_i > 0, \quad \forall i \in \{1, \ldots, M\}, \\
0 < c \leq p \leq p_{\rm{max}}.
\end{cases}
\end{aligned}
\end{equation}

\textbf{3) Stackelberg equilibrium analysis:} Combining Stage II with Stage I, we find the Stackelberg equilibrium to obtain the optimal solution of the formula strategy. In the Stackelberg equilibrium, the leader (i.e., the MASP) will consider the bandwidth demands of users based on the best response to formulate strategies, thus maximizing the MASP utility, and the followers (i.e., users) cannot increase their utility by adjusting their bandwidth requirements. Then the concept of Nash equilibrium is introduced \cite{jiang2022reliable} for the bandwidth demands of the users and the Stackelberg equilibrium for the bandwidth price of the MASP is defined as follows.\\

\begin{definition}[Stackelberg Equilibrium]
We define \( b^*=\{b_i^*\}_{i=1}^M \) as the optimal bandwidth demand strategy vector and \( p^* \) as the optimal unit bandwidth sales price. A Stackelberg equilibrium is achieved if and only if the following two conditions are satisfied:
\[
\begin{cases}
U_{i}\left(b_{i}^{*},b_{-i}^{*},p^{*}\right) \geq U_{i}\left(b_{i},b_{-i}^{*},p^{*}\right),\:\forall i\in\mathcal{M}, \\
U_{s}\left(b^{*},p^{*}\right) \geq U_{s}\left(b^{*},p\right).
\end{cases}
\]
\end{definition}

Below we use backward induction to demonstrate the Stackelberg equilibrium \cite{kang2024tiny}. For clarity without loss of generality, we denote the following $\ln(1+SSIM_{i})$ as $\mathcal{S}_i$.

\newtheorem{theorem}{Theorem}
\begin{theorem}
We first study the existence of the Nash equilibrium of the users' sub-games in Stage II.
\end{theorem}

\begin{proof}
The first-order and second-order derivatives of the utility function $U_i(b_i)$ with respect to the bandwidth demand strategy $b_i$ are
\begin{equation}
\frac{\partial U_i(b_i)}{\partial b_i} =    \frac{\delta_i S_i C_i}{V_i + b_i C_i} - p ,
\end{equation}

\begin{equation}
\frac{\partial^2 U_i(b_i)}{\partial b_i^2} = - \frac{\delta_i S_i C_i^2}{(V_i + b_i C_i)^2} < 0.
\end{equation}

The utility function $U_i(b_i)$ has a first-order derivative equal to zero at a critical point, and its second-order derivative is less than zero, indicating that the optimization problem is concave \cite{kang2024tiny}. To determine the optimal response strategy $b^*$, we set the first-order derivative of $U_i(b_i)$ to zero, giving the optimal response strategy $b^*$ is

\begin{equation} \label{b^*}
b^* =
\begin{cases} 
    \frac{\delta_i \mathcal{S}_i} {p} - \frac{V_i}{C_i}, & \text{if } \frac{\delta_i \mathcal{S}_i C_i}{V_i} > p, \\
    0, & \text{otherwise}.
\end{cases}
\end{equation}

The value of $b^*$ must be greater than zero (i.e., $\frac{\delta_i \mathcal{S}_i C_i}{V_i} > p$), and otherwise the user will not be able to obtain services.
\end{proof}

\newtheorem{theorem1}{Theorem}
\begin{theorem}
With the optimal bandwidth demand response strategy $b^*$, there exists a unique Stackelberg equilibrium $p^*$.
\end{theorem}

\begin{proof} In the Stackelberg game, the MASP, as the leader, is aware that a unique Nash equilibrium exists among users for a given bandwidth price $p$ \cite{zhang2023learning}. Hence, the MASP can maximize its utility by selecting the optimal $p$ through $b^*$ in Stage I. Substituting \eqref{b^*} into \eqref{utilty_MASP}, we have

% 服务器效用函数
\begin{equation}
U_s(p) = \mathcal{P} + \sum_{i=1}^{M} (p-c)\bigg(\frac{\delta_i \mathcal{S}_i}{p}-\frac{V_i}{C_i}\bigg) .
\end{equation}

By calculating the first-order and second-order derivatives of the MASP utility function, we obtain
\begin{equation}
\frac{\partial U_s}{ \partial p} =  \sum_{i=1}^{M}  -\frac{V_i}{C_i} + \frac{c \delta_i \mathcal{S}_i}{p^2},
\end{equation}

\begin{equation}
\frac{\partial^2 U_s(p)}{\partial p^2} =\sum_{i=1}^{M} \frac{-2c  \delta_i \mathcal{S}_i }{p^3} < 0.
\end{equation}

Similarly, setting the first-order derivative of $U_s(p)$ to zero, we can obtain the optimal pricing strategy $p^*$ of the MASP for users as \cite{jiang2022reliable}
\begin{equation}\label{Stackelberg equilibrium}
p^*=\sqrt{\frac{c\sum_{i=1}^{M}  \delta_i \mathcal{S}_i C_i} {\sum_{i=1}^{M} V_i}}.
\end{equation}

The second-order derivative of the utility function $U_s(p)$ of the MASP is negative, indicating its concavity. In addition, the set of pricing strategies of the MASP satisfies the basic criteria of closed, bounded, and convex. Hence, we can uniquely obtain the Stackelberg equilibrium \cite{kang2024tiny}. 
\end{proof}

\section{GDM-based Solution for Optimal Stackleberg Game with Incomplete Information}\label{GDM}
In this section, we model the proposed Stackelberg game between the MASP and users as a 
Partially Observable Markov Decision Process (POMDP) \cite{zhang2023learning}. DRL provides a promising solution to the incentive mechanism under privacy issues \cite{kang2024tiny}. In addition, since traditional DRL algorithms are unstable and difficult to find optimal solutions when solving incomplete information problems \cite{du2023beyond}, we propose a GDM-based DRL algorithm designed to achieve the Stackelberg equilibrium by exploring the solutions for the MASP and users within the Stackelberg game model.

\subsection{POMDP for the Stackelberg Game}
Given the inherent competitive dynamics of the Stackelberg game, each user possesses only partial local information and selects bandwidth strategies in a completely non-cooperative manner \cite{zhang2023learning}. DRL can derive the best strategy from historical experiences, relying only on the current state and the assigned rewards without prior knowledge \cite{Zhang2023Energy}. Therefore, we need to train the DRL agent in an environment that follows the POMDP \cite{kang2024tiny}. The environment is constructed by modeling the dynamic interaction between the MASP and users.

\textit{1) State Space:} In each round of the Stackelberg game, $k\in \{0,\ldots,k,\ldots, K\}$, the MASP agent bases its decisions solely on local observations of the environment. The state space for the MASP in the current game is defined as a combination of its historical pricing strategy from the past $L$ rounds and its current bandwidth demand strategy \cite{zhang2023learning}.
\begin{equation}
\boldsymbol{e}=\{p_{k-L},b_{k-L},p_{k-L+1},b_{k-L+1},\ldots,p_{k-1},b_{k-1}\} .
\end{equation}

In the initial environment, $p_{k-L}$ and $b_{k-L}$ can be randomly generated.

\textit{2) Action Space:} The MASP agent will select an optimal pricing strategy based on the environmental $\boldsymbol{e}$, i.e., historical pricing and bandwidth demand. When receiving an environmental observation, the MASP agent needs to decide on a pricing action $p_k$ to optimize its utility. Therefore, the action space is defined as $\boldsymbol{s}=\{p_k\}$.

\textit{3) Reward:} The MASP receives rewards based on the current environment $\boldsymbol{e}$ and the corresponding actions $\boldsymbol{s}$. We define the POMDP function as the MASP utility function \eqref{utilty_MASP} that we constructed in the Stackelberg game \cite{kang2024tiny}.

\begin{algorithm}[ht]
\DontPrintSemicolon
\SetAlgoLined
    \caption{GDM-based Stackelberg game solution}\label{algorithm}
    \KwIn{Denoising steps $N$, exploration noise $\epsilon$, discount factor $\omega$, batch size $b$, max number of episodes $E_{\rm{max}}$, max number of rounds $K_{\rm{max}}$, max number of steps $S_{\rm{max}}$.}
    \KwOut{optimal bandwidth pricing strategy ${\bm{s}}_0$.}
    
    \textbf{\textit{Phase 1: Initialization}} \\
    Initialize the weights $\theta$ for the Stackelberg game solution  generation network ${\bm{\epsilon}}_\theta$. \\
    Initialize the weights $\upsilon$ for the Stackelberg game solution evaluation network $Q_\upsilon$. \\
    
    \textbf{\textit{Phase 2: Training}} \\
    \For{\text{Episode} $E = 1$ \KwTo $E_{\rm{max}}$}{
        Reset environment state and replay buffer. \\
        \For{\text{Round}  $K = 1$ \KwTo $K_{\rm{max}}$}{
        Initiate bandwidth pricing exploration using a random process ${\mathcal{N}}$. \\
           \For{\text{Step} $S = 1$ \KwTo $S_{\rm{max}}$}{
            Observe the current environmental state $\boldsymbol{e}$. \\
            Set the initial bandwidth pricing strategy ${\bm{s}}_N$ as Gaussian noise. \\
            Generate the bandwidth pricing strategy ${\bm{s}}_0$ by denoising ${\bm{s}}_N$ using the network ${\bm{\epsilon}}_{\theta}$ according to~\eqref{denoise}. \\
            User determine bandwidth demands through \eqref{b^*}
            incorporate exploration noise into ${\bm{s}}_0$. \\
            Implement the generated bandwidth pricing strategy ${\bm{s}}_0$ in the environment and measure the resulting objective function value. \\
            Record the actual objective function value $r\left({\bm{e}},{\bm{s}_0}\right)$. \\
            Update the Stackelberg game solution evaluation network $Q_\upsilon$ according to~\eqref{Q-v}. \\
            Update the Stackelberg game solution generation network ${\bm{\epsilon}}_\theta$ according to~\eqref{epsilon}. \\
        }
        }
    }
    \textbf{return} The trained Stackelberg game solution generation network ${\epsilon}_{\theta}$. \\
    
    \textbf{\textit{Phase 3: Inference}} \\
    Observe the environmental vector ${\boldsymbol{e}}$. \\
    Generate the optimal bandwidth pricing strategy ${\bm{s}}_0$ by denoising Gaussian noise using the trained network ${\bm{\epsilon}}_\theta$ according to \eqref{denoise}. \\
    \textbf{return} The optimal bandwidth pricing strategy ${\bm{s}}_0$.
\end{algorithm}

\subsection{GDM-based Stackelberg Game Solution}
Stackelberg game is a dynamic game model and diffusion model generates the target by denoising multiple time steps \cite{jiang2022reliable,song2020denoising}. The diffusion model combined with the Stackelberg game can fine-tune multiple time steps and enhance the exploration of the Stackelberg equilibrium. Motivated by this characteristic, we propose a GDM-based DRL algorithm that aims to achieve the Stackelberg equilibrium. We map the environmental states to a diffusion model network and define the network as the Stackelberg game policy, i.e., $\pi_{\theta}\left(\boldsymbol{s}|\boldsymbol{e}\right)$, where $\theta$ represents the model parameters \cite{du2023beyond}.

% The goal of $\pi_{\theta}\left(\boldsymbol{s}|\boldsymbol{e}\right)$ is to output a deterministic game design that maximizes the MASP's benefit, thereby maximizing the expected cumulative reward over a series of time steps.
The goal of $\pi_{\theta}\left(\boldsymbol{s}|\boldsymbol{e}\right)$ is to optimize the POMDP output to maximize the expected cumulative reward over a series of time steps, thereby maximizing the benefit of the MASP. We obtain the Stackelberg game policy through the inverse process of the conditional diffusion model \cite{song2020denoising} as 
\begin{equation}\pi_{\theta}(\boldsymbol{s}|\boldsymbol{e})=p_{\theta}\big(\boldsymbol{s}_{0:N}\big|\boldsymbol{e}\big)=\mathcal{N}(\boldsymbol{s}_{N};\boldsymbol{0},\mathbf{I})\prod_{n=1}^{N}p_{\theta}\big(\boldsymbol{s}_{n-1}\big|\boldsymbol{s}_{n},\boldsymbol{e}\big).
\end{equation}

The training process of GDMs involves optimizing the negative log-likelihood of the training data \cite{du2023beyond}. In the denoising process, historical user bandwidth demands and the MASP pricing condition information, i.e., $\boldsymbol{e}$, are added. $p_\theta(\boldsymbol{s}_{n-1}|\boldsymbol{s}_n,\boldsymbol{e})$ can be interpreted as building a noise prediction model. The covariance matrix is expressed as \begin{equation}\Sigma_\theta\left(\boldsymbol{s}_n,\boldsymbol{e},n\right)=\beta_n\mathbf{I}.\end{equation}

The mean is designed as \begin{equation} \mu_{\theta}\left(\boldsymbol{s}_{n},\boldsymbol{e},n\right)=\frac{1}{\sqrt{\alpha_{n}}}\left(\boldsymbol{s}_{n}-\frac{\beta_{n}}{\sqrt{1-\bar{\alpha}_{n}}}\boldsymbol{\epsilon}_{\theta}\left(\boldsymbol{s}_{n},\boldsymbol{e},n\right)\right).
\end{equation}

The MASP pricing strategy $\boldsymbol{s} \sim \mathcal{N}(\boldsymbol{0},\mathbf{I})$ can be sampled from the reverse forward diffusion process \cite{du2023beyond}, which is parameterized by $\theta$ as
\begin{equation}\label{denoise}
\boldsymbol{s}_{n-1}\mid \boldsymbol{s}_n=\frac{\boldsymbol{s}_n}{\sqrt{\alpha_n}}-\frac{\beta_n}{\sqrt{\alpha_n\left(1-\bar{\alpha}_n\right)}}\epsilon_\theta\left(\boldsymbol{s}_n,\boldsymbol{e},n\right)+\sqrt{\beta_n}\epsilon,\end{equation}
where $\epsilon \sim \mathcal{N}(0,I)$ and $n=1,\ldots,N$.

Motivated by the Q-function in DRL, we introduce a Stackelberg game solution evaluation network $Q_{v}$, which can assign a Q-value to a pair of bandwidth pricing policies (i.e., $\boldsymbol{e}$ and $\boldsymbol{s}$), which represents the expected cumulative reward when the agent takes a Stackelberg game policy from the current state and follows the policy. The $Q_{v}$ network acts as a navigation aid during the training of the GDM network, i.e., the Stackelberg game solution generation network $\epsilon_{\theta}$. The optimal $\epsilon_{\theta}$ is the Stackelberg game solution generation network with the highest expected Q-value. Therefore, the loss function of the network can be expressed as \cite{du2023beyond}
\begin{equation}\label{epsilon}
\arg\min_{\boldsymbol{\epsilon}_\theta}\mathcal{L}_{\boldsymbol{\epsilon}}(\theta)=-\mathbb{E}_{\boldsymbol{s}_0\sim\boldsymbol{\epsilon}_\theta}\left[Q_v\left(\boldsymbol{e},\boldsymbol{s}_0\right)\right].\end{equation}

We aim to train the Stackelberg game solution evaluation network $Q_{v}$ so that the difference between the predicted Q-value and the actual Q-value is minimized. Hence, the optimization of $Q_{v}$ is
\begin{equation}\label{Q-v}
\arg\min_{Q_v}\mathcal{L}_Q(v)=\mathbb{E}_{\boldsymbol{s}_0\sim\pi_\theta}\left[\left\|r(\boldsymbol{e},\boldsymbol{s}_0)-Q_v\left(\boldsymbol{e},\boldsymbol{s}_0\right)\right\|^2\right],\end{equation}
where $r$ represents the objective function value of implementing $\boldsymbol{s}_0$  under $\boldsymbol{e}$.

The overall algorithm of the GDM-based Stackelberg game solution is shown in Algorithm~\ref{algorithm}. $\psi_g$ and $\psi_e$ are the weights in the Stackelberg game solution generation network and the evaluation network, respectively. Each update of the generation network and evaluation network results in the complexity of $\mathcal{O}(\psi_g )$ and $\mathcal{O}(\psi_e)$ respectively. Replay buffer activities maintain a storage complexity of $\mathcal{O}(1)$. The computational complexity of the training phase is \( \mathcal{O}(E_{\text{max}}K_{\text{max}} S_{\text{max}} (N \psi_g + \psi_e )\). In the inference phase, the complexity of generating the MASP pricing policy using the trained network is \( \mathcal{O}(\psi_e )\), assuming that reward observation and exploration noise generation are constant time operations.

\begin{figure*}[t]
\centering
\subfigure[SSIM.]
{
    \begin{minipage}[t]{0.4\linewidth}
	\centering
	\includegraphics[width=1\linewidth]{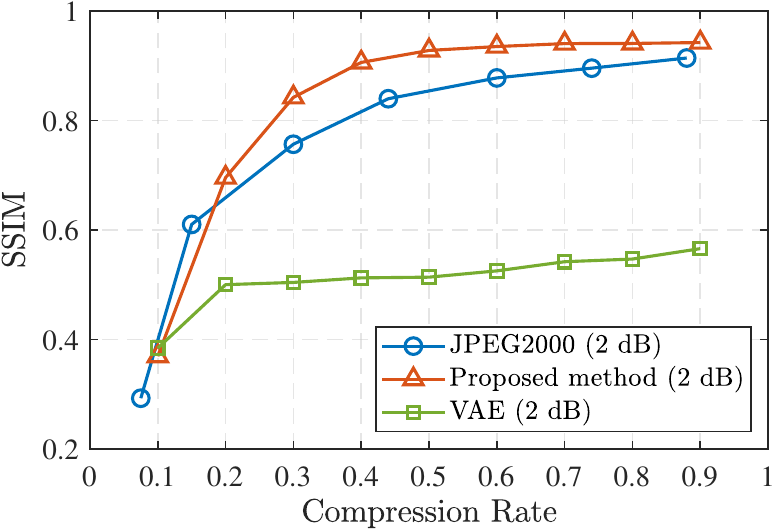}
        \captionsetup{font=footnotesize}
	\label{ssim_snr2}
    \end{minipage}
}
\hspace{0.3in}
\subfigure[PSNR.]
{
    \begin{minipage}[t]{0.4\linewidth}
	\centering
	\includegraphics[width=1\linewidth]{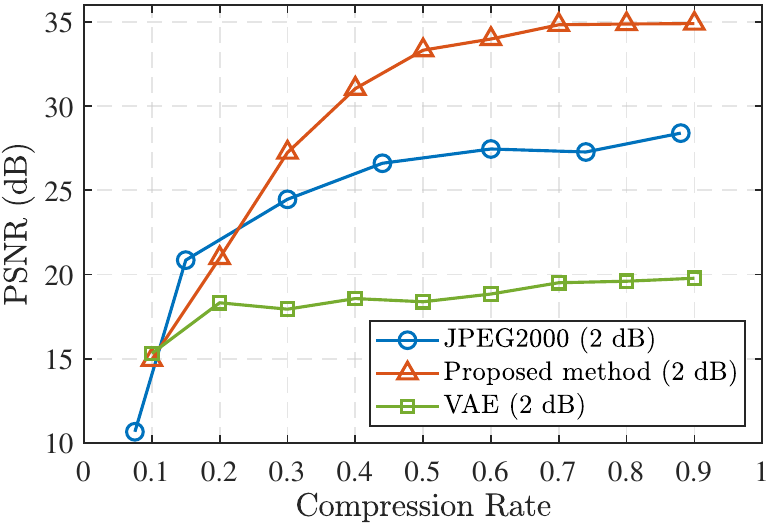}
        \captionsetup{font=footnotesize}
	\label{psnr_snr2}
    \end{minipage}
}
\captionsetup{font=footnotesize}
\caption{Performance comparison of the proposed method with other compression algorithms in a channel with SNR = 2 dB.}\label{snr2}
\end{figure*}

\begin{figure*}[t]
\centering
\subfigure[SSIM.]
{
    \begin{minipage}[t]{0.4\linewidth}
	\centering
	\includegraphics[width=1\linewidth]{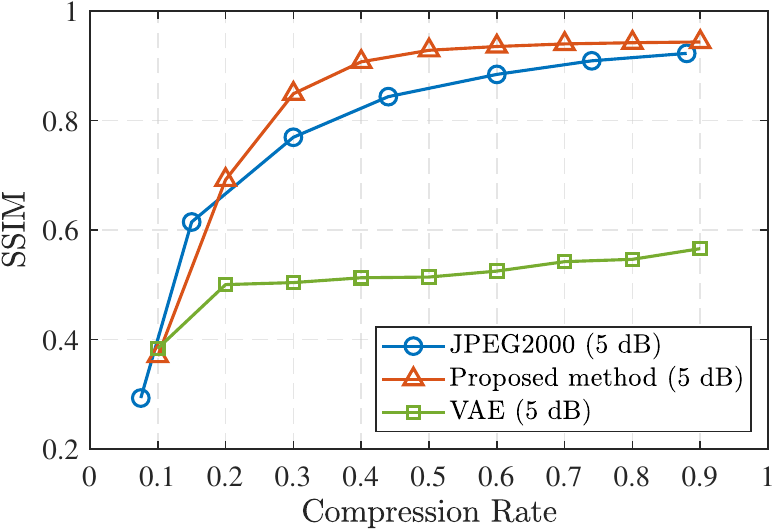}
        \captionsetup{font=footnotesize}
	\label{ssim_snr5}
    \end{minipage}
}
\hspace{0.3in}
\subfigure[PSNR.]
{
    \begin{minipage}[t]{0.4\linewidth}
	\centering
	\includegraphics[width=1\linewidth]{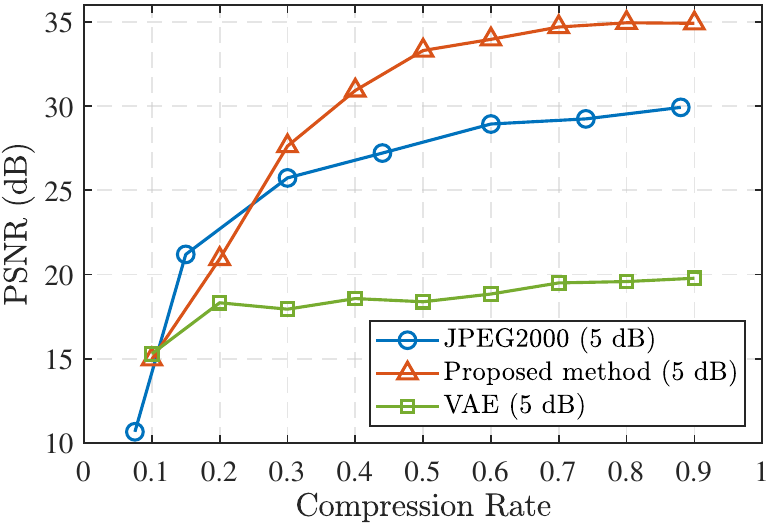}
        \captionsetup{font=footnotesize}
	\label{psnr_snr5}
    \end{minipage}
}
\captionsetup{font=footnotesize}
\caption{Performance comparison of the proposed method with other compression algorithms in a channel with SNR = 5 dB.}\label{snr5}
\end{figure*}

\begin{figure*}[t]
    \begin{center}
	\begin{minipage}[t]{0.4\linewidth}
		\centering
            \captionsetup{font=footnotesize}
            \includegraphics[width=1\textwidth]{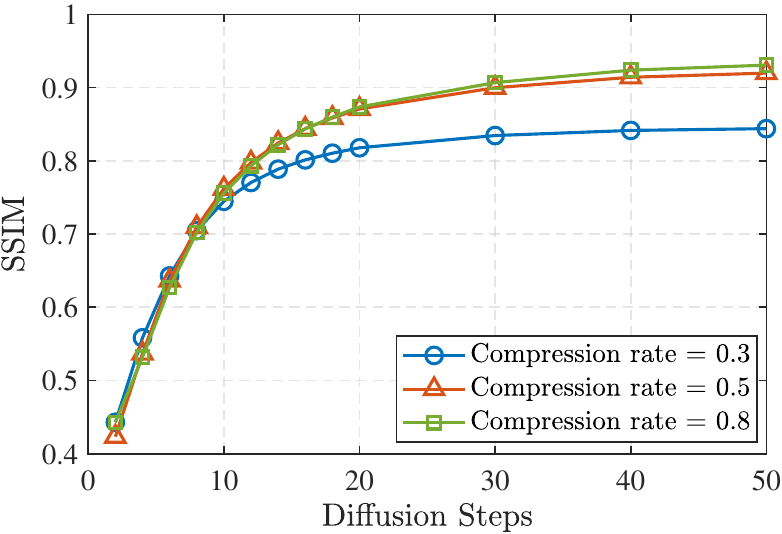}
            \caption{SSIM corresponds to different diffusion steps under different compression rates.}\label{SSIM_step}
	\end{minipage}
	\hspace{0.3in}
	\begin{minipage}[t]{0.4\linewidth}
		\centering
            \includegraphics[width=1\linewidth]{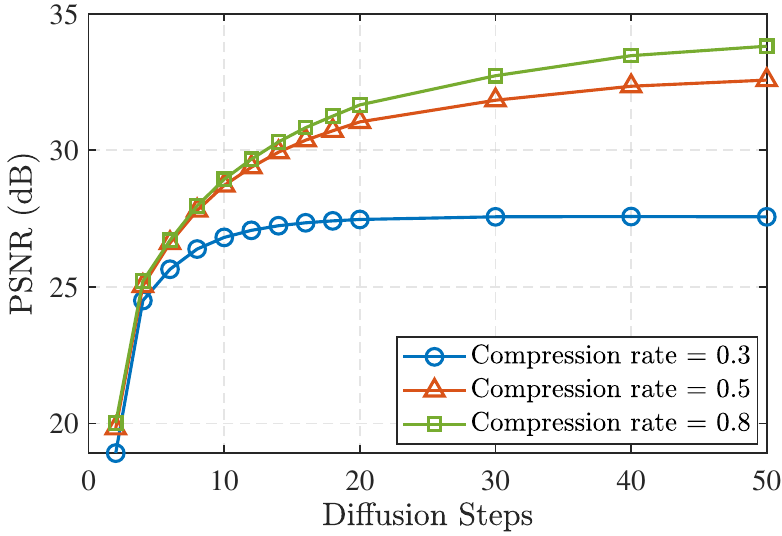}
		\captionsetup{font=footnotesize}
            \caption{PSNR corresponding to different diffusion steps under different compression rates.}\label{PSNR_step}
	\end{minipage}
	\end{center}
\end{figure*}

\begin{figure*}[t]
%\vspace{-0.5cm}
\centering
\includegraphics[width=0.9\textwidth]{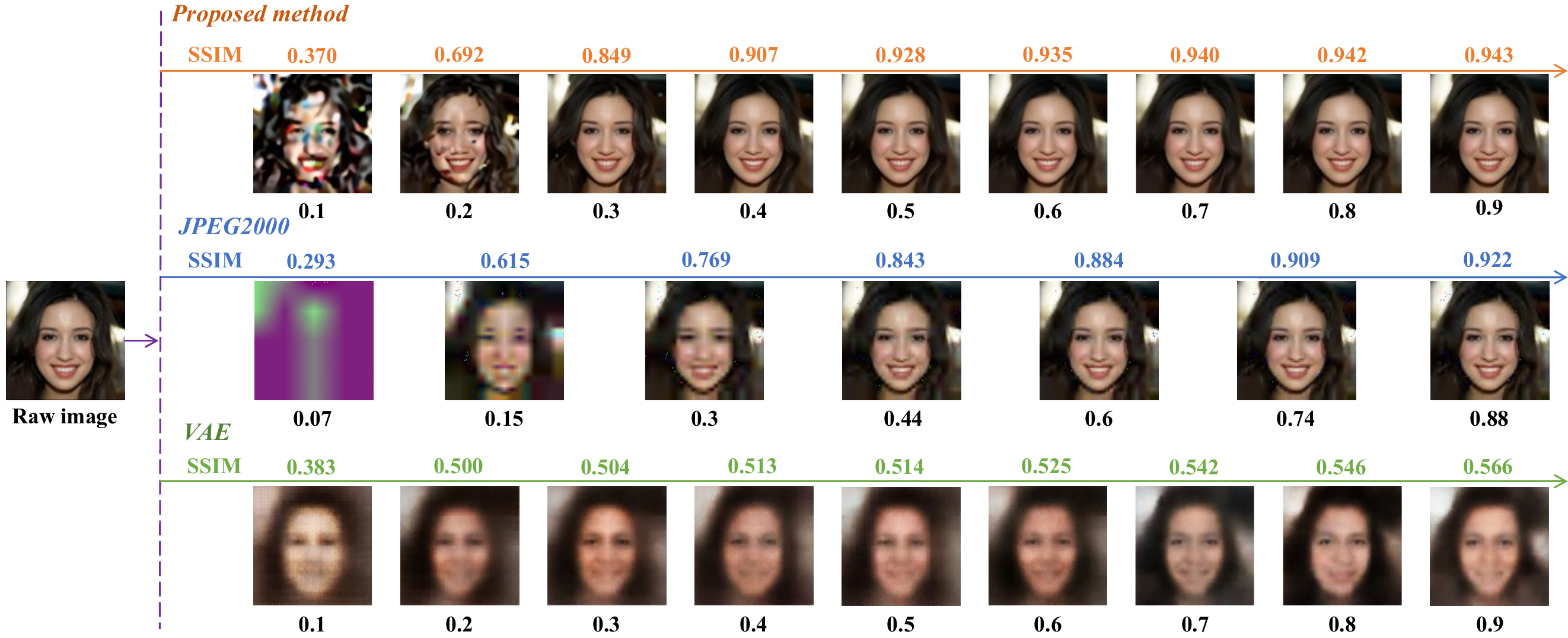}
\captionsetup{font=footnotesize}
\caption{Raw image and images recovered by different methods under different compression rates in a channel with SNR = 5 $dB$.}
\label{example}
\end{figure*}

\begin{table}[t]
\label{parameter}
\renewcommand{\arraystretch}{1}
\captionsetup{font = small}
\centering
\caption{Summary of Training Hyperparameters \cite{du2023beyond}}
\begin{tabular}{>{\centering\arraybackslash}p{0.8cm}>{\raggedright\arraybackslash}p{4.7cm}>{\centering\arraybackslash}p{1.5cm}} \toprule[1.2pt]
 \textbf{Symbols} & \textbf{Parameters} & \textbf{Values} \\ \midrule
$\lambda$&Weight decay &$1\times10^{-4}$\\
$B$&  Batch size&$512$\\
$N$&  Denoising steps&$5$\\
$D$&  Maximum capacity of the replay buffer&$1\times10^{6}$\\
$\eta_{\mathrm{a}}$&  Learning rate of the generation networks&$1\times10^{-4}$\\
$\eta_{\mathrm{c}}$&  Learning rate of the evaluation networks&$1\times10^{-3}$\\
$\tau$& Weight of soft update&$0.005$\\
$\omega$& Discount factor &$0.95$\\\bottomrule[1.2pt]
\end{tabular}\label{parameter}
\end{table}

\begin{figure}
    \centering
    \includegraphics[width=0.4\textwidth]{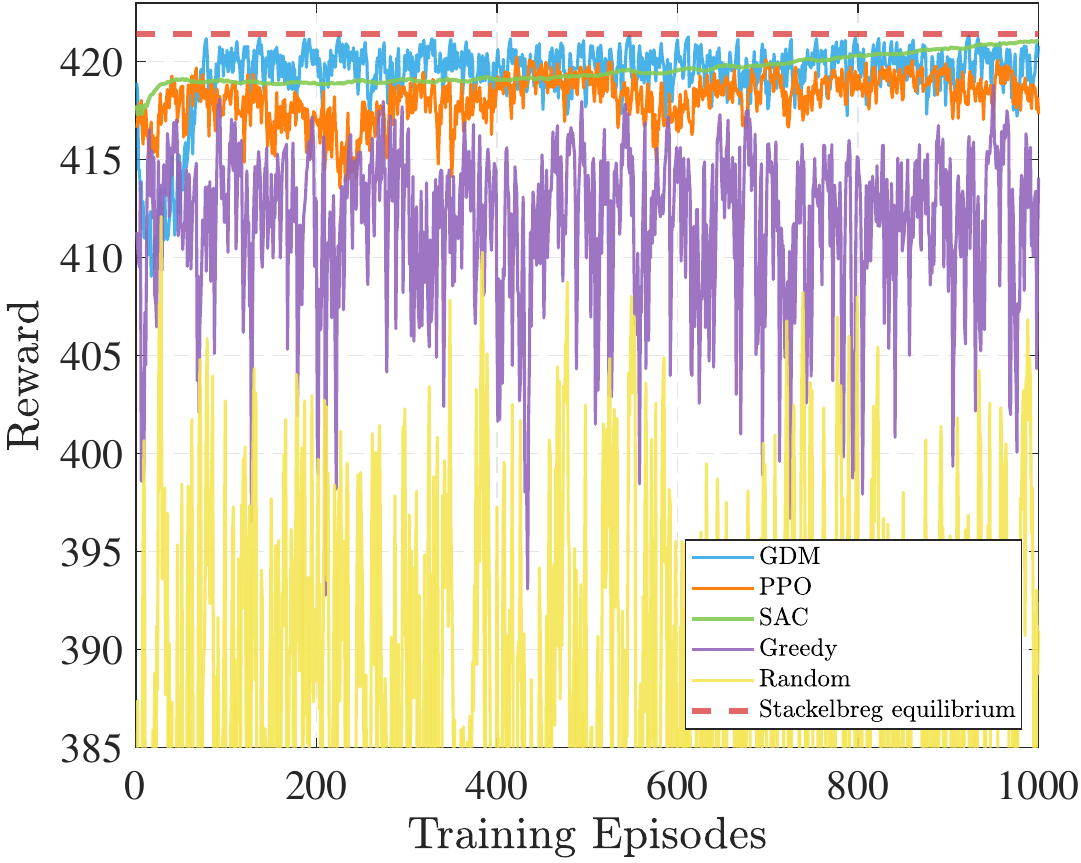}
    \captionsetup{font=footnotesize}
    \caption{The reward comparison of proposed GDM-based DRL algorithm with other algorithms, i.e., PPO, SAC, Greedy, and Random.}
    \label{compare_average_random}
\end{figure}

\begin{figure*}
    \centering
    \begin{minipage}{0.4\linewidth}
        \centering
        \includegraphics[width=1\textwidth]{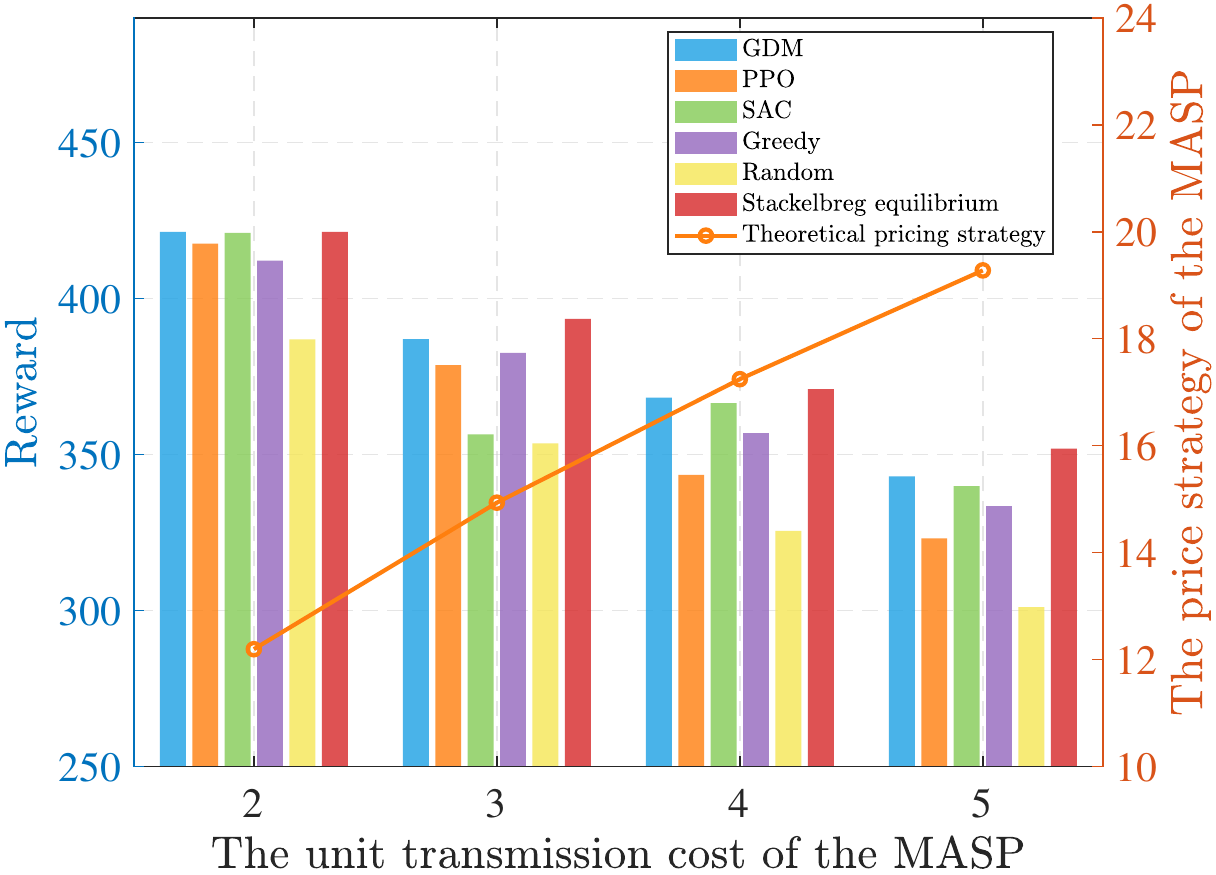}
        \captionsetup{font=footnotesize}
        \caption{The performance of the proposed GDM-based DRL algorithm is evaluated as the unit transmission cost varies within $[2, 5]$.}
        \label{transmission_cost}
    \end{minipage}
    \hspace{0.7in}
    \begin{minipage}{0.4\linewidth}
        \centering
        \includegraphics[width=1\textwidth]{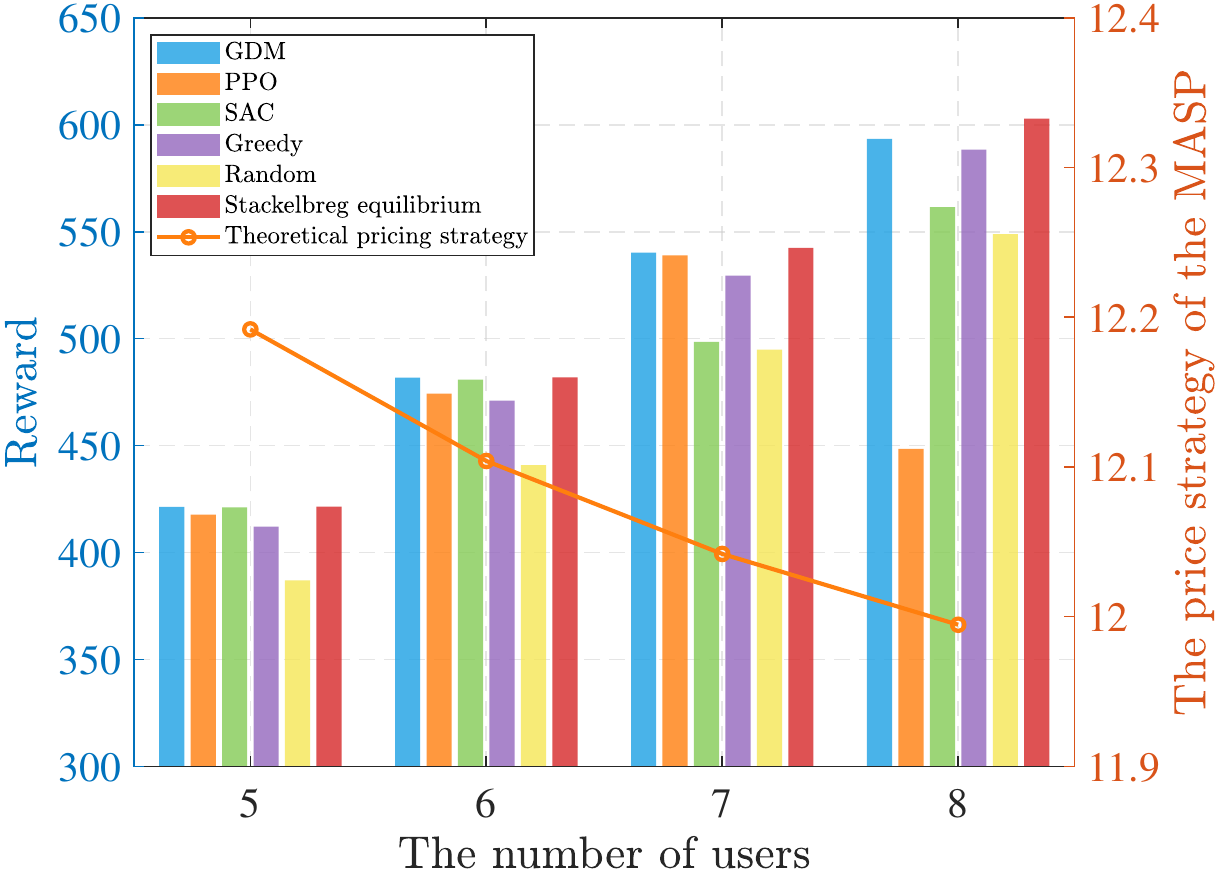}
        \captionsetup{font=footnotesize}
        \caption{The performance of the proposed GDM-based DRL algorithm is evaluated as the number of users varies within $[5, 8]$. }
        \label{number_users}
    \end{minipage}
\end{figure*}

\begin{figure}
    \centering
    \includegraphics[width=0.4\textwidth]{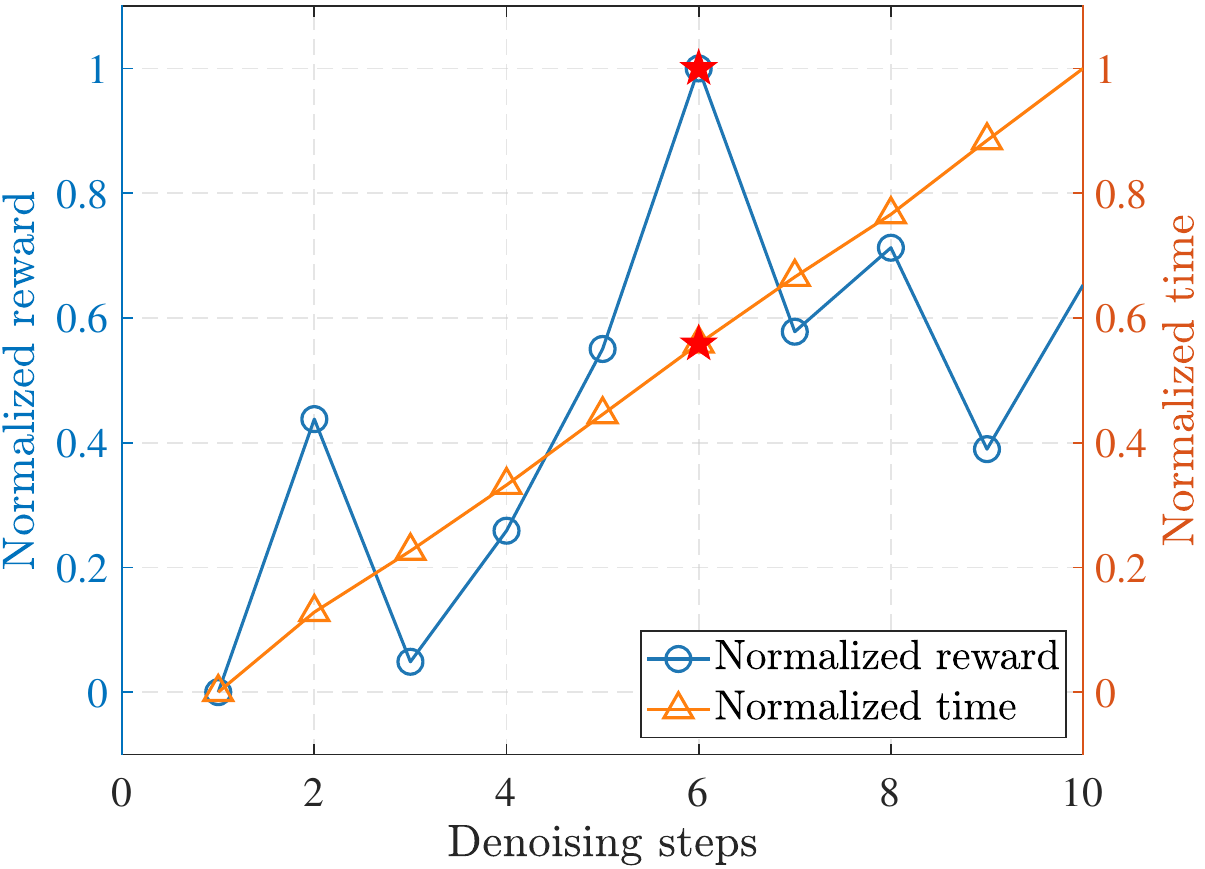}
    \captionsetup{font=footnotesize}
    \caption{The values of normalized reward and normalized time correspond to different denoising steps.}
    \label{denoising_value}
\end{figure}

\section{Numerical Results}\label{results}
% Based on our core contribution, this section will verify the effectiveness of our proposed GM-SemCom and GDM-based Solution for Stackelberg Game through experiments.

% \begin{description}
% \item[\textbf{Q1)}]Can the proposed GM-SemCom effectively save communication resources, overcome the dynamic inference characteristics of GAI technology, and realize accurate information transmission?

% \item[\textbf{Q2)}]Can the proposed Stackelberg game model maximize the benefits of the MASP while providing users with high-quality AIGC services?

% \item[\textbf{Q3)}]Can the proposed GDM-based DRL algorithm effectively solve the Stackelberg game with incomplete information between the MASP and users? 
% \end{description}
\subsection{Simulation Settings}
We verify the effectiveness of the proposed GM-SemCom in mobile AIGC networks with OFDMA channels~\eqref{OFDMA} with different quality channel conditions. In the experiment, we consider two benchmark schemes. JPEG2000 Scheme: utilizes traditional Separate Source-Channel Coding (SSCC), employing JPEG2000 for image compression \cite{christopoulos2000jpeg2000}. VAE Scheme: VAE-based SemCom combined with Kullback-Leibler divergence loss\cite{kingma2013auto}. The training dataset is the FFHQ dataset with image sizes of $128 \times 128$ \cite{karras2019style}, and the test dataset is a subset of the CELEBA dataset of the same size \cite{Liu2015celebA}. Our experiments are a re-implementation of DDIM, based on the improved architecture by \cite{du2024generative,dhariwal2021diffusion}.

The parameter settings of the GDM-based DRL algorithm for solving the Stackelberg game are shown in Table~\ref{parameter}. We consider a scenario where there are $M=5$ users. The unit transmission cost $c$ is $2$, the AIGC service cost $\mathcal{P}$ is $10$, the maximum bandwidth is $200~\text{MHz}$, the maximum bandwidth price is $20$, and the immersion coefficient $\delta_i$ is $15$. The transmitter power of the MASP is $40~\text{dBm}$, the unit channel power gain is $-20~\text{dB}$, the path loss exponent is $2$, and the average noise power $N_0$ is $-150~\text{dBm}$ \cite{zhang2023learning}. Depending on the dataset, the bit volume of the transmitted AIGC product is $\mathcal{L}_s = 10~\text{MB}$. The distance between MASP and the five users is $[100, 200, 300, 400, 500]$ $m$, the compression rate is $[0.3, 0.4, 0.5, 0.6, 0.7]$, and the SSIM of the AIGC product received by the users is $[0.75, 0.8, 0.85, 0.9, 0.95]$. We set $E = 1000$ and $K = 10$ during the experiment.

\subsection{Effectiveness of GM-SemCom } \label{effetiveness of GM}
% \subsection{Effectiveness of GM-SemCom (for \textbf{Q1})} \label{effetiveness of GM}
The degree of visual distortion between the source image and the reconstructed image is measured using SSIM~\eqref{SSIM} and PSNR \cite{hore2010image}. The PSNR is defined as
\begin{equation}PSNR(x,y)=10\cdot\log_{10}\bigg(\frac{MAX^{2}}{MSE}\bigg),\end{equation}
where $MSE$ represent the mean squared error between the source image $x$ and the reconstructed image $y$, and the $MAX$ represents the maximum absolute error of image pixels \cite{hore2010image}. The $MSE$ is defined as
\begin{equation}MSE=\frac{1}{H\times W}\sum_{i=1}^{H}\sum_{i=1}^{W}\left(x(i,j)-y(i,j)\right)^{2},\end{equation}
where \( x(i,j) \) and \( y(i,j) \) denote the pixel values at the corresponding coordinates, \( H \) and \( W \) represent the height and width of the images, respectively. A smaller $MSE$ indicates a closer resemblance between two images, signifying reduced distortion.

Fig. \ref{snr2} and \ref{snr5} present a comparative analysis of the proposed GM-SemCom against established baseline schemes, evaluating SSIM and PSNR as functions of the compression rate under various channel qualities. The results indicate that our proposed method achieves superior SSIM and PSNR values at compression rates above approximately $0.3$ and maintains commendable performance even under poor channel conditions. In contrast, the VAE-based method exhibits the lowest performance at lower compression rate settings. This is due to the fact that the VAE-based model cannot effectively reconstruct source information when feature compression is too high. Additionally, the dynamic inference characteristics of the single-modal GAI model may impede accurate information recovery. These findings underscore the capability of the proposed GM-SemCom to enhance the accuracy of information transmission while minimizing resource consumption.

The performance of the decoder in the GM-SemCom is further evaluated under the different compression rate and diffusion steps, as shown in Fig. \ref{SSIM_step} and Fig. \ref{PSNR_step}. The experimental results show that increasing the compression rate and diffusion steps can reduce the quality of image reconstruction. Once reconstruction quality is high, further increasing the compression rate and diffusion steps yields minimal improvement. Therefore, it is important to find a balance between resource consumption and service quality to optimally meet user needs. These findings highlight the need for precise adjustment of GM-SemCom parameters to ensure an effective combination of resource consumption and service quality. As shown in Fig. \ref{example}, we present an example of using the GM-SemCom for data transfer at different compression rates in the test dataset and compare it with the baseline method. Numerical results demonstrate that the accuracy and clarity of information transmitted by GM-SemCom are improved under the same compression rate. This confirms the robustness of the proposed GM-SemCom in energy-efficient information transmission and its effectiveness in integrated data communication.
%\subsection{Effectiveness of the proposed Stackelberg game (for \textbf{Q2})}

\subsection{Effectiveness of GDM-based Solution for Stackelberg Game }
% \subsection{Effectiveness of GDM-based Solution for Stackelberg Game (for \textbf{Q2 and Q3})}

Firstly, we demonstrate the performance of our proposed GDM-based DRL algorithm. We compare the proposed GDM-based DRL algorithm ($N=5$) with two traditional algorithms (i.e., random and greedy) and two standard DRL algorithms (i.e., SAC and PPO). As shown in Fig. \ref{compare_average_random}, as the number of episodes increases, we find that only the proposed GDM-based DRL algorithm and the SAC algorithm can reach convergence, while the GDM-based DRL algorithm converges faster. In addition, compared with the random algorithm, greedy algorithm, SAC algorithm, and PPO algorithm, the final result of the GDM-based DRL algorithm can almost reach the Stackelberg equilibrium even without complete information, allowing the MASP to reach the theoretical optimal payoff. The GDM-based DRL algorithm utilizes dynamic inference characteristics of the diffusion model and has strong exploration capabilities. The proposed algorithm of the multi-step generation strategy is consistent with the Stackelberg game under incomplete information. Note that the Stackelberg equilibrium is obtained by \eqref{Stackelberg equilibrium} proposed in Section \ref{Stackelberg Game between MAPS and users}.

In Fig. \ref{transmission_cost}, we study the impact of transmission cost on the MASP. Specifically, we investigate the impact of unit transmission cost by gradually varying the cost from $2$ to $5$. It can be seen that the proposed GDM-based DRL algorithm makes the utility of the MASP close to the Stackelberg equilibrium, and its utility is significantly better than the baseline. This result shows that the GDM-based DRL algorithm can find the optimal solution under incomplete information. As the cost increases, the revenue of the MASP decreases. This is because the increase in the cost also increases the pricing of the MASP, thereby reducing the user's demand for purchasing bandwidth, resulting in a decrease in the utility of the MASP.

Next, we study the impact of the MASP serving different numbers of users. We set the unit transmission cost to $5$ and gradually increased the number of users from $5$ to $8$. The parameters for the new users are uniformly set to a distance of $500$ $\rm{m}$, compression rate of $0.5$, and SSIM of $0.8$. As shown in Fig. \ref{number_users}, as the number of users increases, the utility of the MASP also increases. The reason is that when there are fewer users, the MASP has sufficient bandwidth to serve all users. The MASP reduces pricing to encourage users to purchase more bandwidth, which increases utility. However, when the number of users is larger, the MASP’s bandwidth may become insufficient. Consequently, the MASP increases the bandwidth price to restrict excessive purchases and maintain service quality.

In the GDM-based DRL algorithm, increasing diffusion steps $N$ enhances exploration and rewards beyond local optima, benefiting the MASP. However, more steps are not always better as they increase time consumption \cite{du2024diffusion}. In our example, resource pricing strategies with over six steps struggle to converge, and training becomes inefficient. As shown in Fig. \ref{denoising_value}, we vary the denoising step from $N=1$ to $N=10$. We observe that as the number of denoising steps increases, the reward first increases and then decreases, while the training time continues to increase. This demonstrates that the reward curve has an optimal number of denoising steps at the inflection point. This is because too many diffusion steps will weaken the ability to effectively explore the reward, while excessive denoising would lead to overfitting. Thus, balancing learning performance and computational efficiency to choose the appropriate denoising step is key to achieving optimal rewards.

\section{Conclusion}\label{conclusion}
% Conclusion + future work
In this paper, we have proposed a GM-SemCom framework for mobile AIGC networks to achieve high-quality and low-latency AIGC services. Specifically, we have integrated diffusion models with multi-modal semantic information and designed a controllable compression module to achieve controllable semantic information compression. Numerical results have clearly demonstrated that the performance of GM-SemCom is better than those of the  JPEG2000 and VAE-based SemCom. To enhance the QoE of AIGC services, we have introduced a new metric called AoSI. Additionally, we have proposed a Stackelberg game model that incorporates AoSI and psychological factors into the mechanism design to address the resource transaction problem within the framework. Furthermore, we have designed a GDM-based DRL algorithm to solve the Stackelberg game under conditions of incomplete information. Compared with traditional DRL, our algorithm achieves an optimal scheme that converges faster and is closer to the Stackelberg equilibrium. For future work, it is crucial to enhance security measures against known vulnerabilities to ensure the reliability and safety of AIGC services \cite{kang2023adversarial}. As AIGC services expand, prioritizing robust privacy mechanisms is also essential to mitigate data breach risks.

\bibliographystyle{IEEEtran}

\bibliography{ref}
\end{document}